\documentclass[letterpaper]{article} %
\usepackage{aaai23}  %
\usepackage{times}  %
\usepackage{helvet}  %
\usepackage{courier}  %
\usepackage[hyphens]{url}  %
\usepackage{graphicx} %
\usepackage{natbib}  %
\usepackage{caption} %
\usepackage{amsmath}
\usepackage{amsfonts}
\usepackage{amssymb}
\usepackage{amsthm}
\usepackage{mathtools}
\usepackage{booktabs}
\usepackage{multirow,array}

\makeatletter
\newcommand{\removelatexerror}{\let\@latex@error\@gobble}
\makeatother

\usepackage{paralist}
\usepackage{scrextend}          %

\usepackage[sort&compress,nameinlink,noabbrev,capitalize]{cleveref}

\usepackage{needspace}

\usepackage[noend,ruled,linesnumbered]{algorithm2e}
\SetAlFnt{\small}
\SetAlCapFnt{\small}
\SetAlCapNameFnt{\small} 
\SetAlCapHSkip{0pt}

\SetKw{KwAnd}{and}
\SetKw{KwOr}{or}
\SetKw{KwNo}{no}
\SetKw{KwYes}{yes}
\SetKw{KwT}{true}
\SetKw{KwF}{false}
\SetKw{KwST}{s.t.}
\SetKwInOut{Input}{input}
\SetKwInOut{Output}{output}
\SetKwProg{Def}{def}{:}{}

\SetKwComment{Mycomment}{$\triangleright$\ }{}

\usepackage[textsize=tiny,textwidth=1.5cm,linecolor=green!70!black, backgroundcolor=green!10, bordercolor=black,disable]{todonotes}

\newcommand{\todoH}[1]{\todo[linecolor=yellow!70!black, backgroundcolor=yellow!10]{H: #1}}

\newcommand{\todoS}[1]{\todo[linecolor=blue!70!black, backgroundcolor=blue!10]{S: #1}}

\usepackage[final,notref,notcite]{showkeys}
\newcommand{\bigO}{\ensuremath{\operatorname{O}}}%

\theoremstyle{plain}
\newtheorem{theorem}{Theorem}[section]
\newtheorem{lemma}[theorem]{Lemma}

\newtheorem{observation}[theorem]{Observation}

\newtheorem{claim}[theorem]{Claim}

\theoremstyle{definition}
\newtheorem{definition}[theorem]{Definition}

\newcommand{\prob}[6]{%
 \begin{center}%
   \begin{minipage}{0.9\linewidth}%
     \begin{itemize}[d]
       \item[\textsc{#1}]
       \item[\textbf{#2}]  #3%
       \item[\textbf{#4:}]  #5
     \end{itemize}
    \end{minipage}%
   \end{center}
}

\newcommand{\probdef}[3]{\prob{#1}{Instance:}{#2}{Question:}{#3}{as}}

\newcommand{\cost}{\ensuremath{\operatorname{cost}}}
\newcommand{\pEHG}{\textsc{Exact Game Implementation}}
\newcommand{\pHG}{\textsc{Game Implementation}}

\newcommand{\pXTC}{\textsc{Exact Cover by 3-Sets}}
\newcommand{\pTC}{\textsc{3-Coloring}}

\newcommand{\equib}{equitable}

\newcommand{\ccc}{\mathcal{C}}
\newcommand{\aaa}{\mathcal{A}}
\newcommand{\sss}{\mathcal{S}}
\newcommand{\vvv}{\mathcal{V}}
\newcommand{\uuu}{\mathcal{U}}
\newcommand{\xxx}{\mathcal{X}}
\newcommand{\ooo}{\mathcal{O}}
\newcommand{\fff}{\mathcal{F}}
\newcommand{\enn}{\ensuremath{\hat{n}}}

\DeclareMathOperator{\ngb}{N}

\newif\ifshort
\shortfalse 

\newcommand{\appsymb}{\ensuremath{\star}}

\ifshort %

\usepackage{etoolbox} %
\newcommand{\toappendix}[1]{%
  \gappto{\appendixtext}{
    {#1}
   }
}

\newcommand{\appendixproof}[2]{%
  \gappto{\appendixtext}{
    \subsection{Proof of \cref{#1}}\label{proof:#1}
    #2
    }
}

\newcommand{\appendixsection}[1]{%
  \gappto{\appendixtext}{
    \section{Additional Material for Section~\ref{#1}}
    \label{appsec:#1}
  }
}

\else %

\usepackage{etoolbox} %
\newcommand{\toappendix}[1]{#1}

\newcommand{\appendixproof}[2]{#2}

\newcommand{\appendixsection}[1]{}

\fi

\title{Game Implementation: What Are the Obstructions?}
\author {
    Jiehua Chen,
    Sebastian Vincent Haydn,
    Negar Layegh Khavidaki,
    Sofia Simola,
    Manuel Sorge
}
\affiliations {
    TU Wien\\
    \{jiehua.chen, seyedehngar.khavidaki, sofia.simola, manuel.sorge\}@tuwien.ac.at,\\
    e51832021@student.tuwien.ac.at
}
\ifshort \else
\nocopyright
\fi

\begin{document}

\maketitle
\begin{abstract}
  In many applications, we want to influence the decisions of independent agents by designing incentives for their actions.
  We revisit a fundamental problem in this area, called \pHG: Given a game in standard form and a set of desired strategies, can we design a set of payment promises such that if the players take the payment promises into account, then all undominated strategies are desired?
  Furthermore, we aim to minimize the cost, that is, the worst-case amount of payments.

  We study the tractability of computing such payment promises and determine %
  more closely what obstructions we may have to overcome in doing so.
  We show that \pHG\ is NP-hard even for two players, solving in particular a long open question (Eidenbenz et al.\ 2011) and suggesting more restrictions are necessary to obtain tractability results.
  We thus study the regime in which players have only a small constant number of strategies and obtain the following.
  First, this case remains NP-hard even if each player's utility depends only on three others.
  Second, we repair a flawed efficient algorithm for the case of both small number of strategies and small number of players.
  Among further results, we characterize sets of desired strategies that can be implemented at zero cost as a kind of stable core of the game.
\end{abstract}

\section{Introduction}

Nudge theory~\cite{thaler_nudge_2008}, gamification~\cite{hamari_gamification_2019}, and the design of blockchain systems~\cite{ButerinRLP20} are just a few areas in which we apply incentives in order to coax agents towards behaving in a desirable way.
In these general settings, agents select strategies on their own volition, but we may add incentives (or incur penalties) that increase (resp.\ decrease) the salience or utility of particular strategies in situations of our choice.
The goal is to implement a desired set of strategies or strategy profiles, that is, to ensure that undesired strategies entail smaller utility than desired ones.

With the advent of blockchain systems, we feel that this topic has gained renewed relevance.
First, the design of a blockchain system itself, such as Bitcoin or Ethereum, involves the design of a protocol that rewards intended behavior (e.g., validating transactions by mining blocks for block rewards in Bitcoin) or penalizes unintended behavior (e.g., by slashing the stake of validators that deviate from a consensus in the recent upgrade of Ethereum).
The latter is a form of enforcing the existence of a Schelling point via incentives.
Second, there are now base-layer systems like Ethereum in place that allow world-wide consistent general-purpose computations and thus the straightforward creation of new moneys (called tokens) that can be made to behave in new ways: generated, burned, exchanged, locked, etc.
Thus an immense design space for incentive-based protocols was opened up and we witness its continued exploration.
For just a few examples consider stablecoins, that is, tokens that use incentive-based mechanisms to try and reflect the value of some underlying security (Maker DAI, Terra USD, FRAX Shares, and many more)\footnote{See \url{https://makerdao.com/en/whitepaper/}, \url{https://terra.money/Terra_White_paper.pdf}, and \url{https://docs.frax.finance/overview}.} or incentive-based consensus mechanisms for adjudication, moderation, and transferring real-world information onto blockchains (such as Kleros, UMA, Chainlink oracles, and again many more)\footnote{See \url{https://kleros.gitbook.io/docs/}, \url{https://docs.umaproject.org/}, and \url{https://chain.link/whitepaper}.}.

In all of the above design problems, there are independent actors that we want to incentivize to behave in a certain desired way.
A fundamental underlying problem herein is \pHG\ \cite{monderer_kimplementation_2004}, stated as follows:
We are given a game in standard form (a set of players, strategies for each player, and their utility) and for each player a set of desired strategies.
We want to specify a set of payment promises that define for each strategy profile (a tuple specifying one strategy per player) a payment promise to each player.
These payment promises shall \emph{implement} our desired sets of strategies, that is, when taking the payments into account, no player wants to play an undesired strategy.
In technical terms, each strategy that is not dominated by any other strategy is desired (see \cref{sec:prelims} for the formal definitions)%
.\footnote{We focus here only on pure strategies. Furthermore, \pHG\ indeed is to implement a set of strategy profiles rather than a set of strategies for each player. Implementing sets of strategies corresponds to implementing so-called rectangular strategy profiles, see the formal definitions in \cref{sec:prelims}.}
Furthermore, we want to minimize the \emph{cost} of the implementation, that is, the amount paid in the worst case.
More precisely, we want to minimize, over all strategy profiles that consist of undominated strategies, the sum of payment promises to all players.
In this work, we explore the question ``How difficult is it to implement a desired set of strategies?''

\paragraph{Contribution.}
We obtain the following results.
We first show that \pHG\ is NP-hard, even if there are only two players and even if our budget for the cost is 0 (\cref{thm:2p}).
This strengthens two results by \citet{deng_complexity_2016} who showed that \pHG\ is NP-hard for six players, and that \pHG\ is NP-hard for two players and mixed strategies, both with positive budgets.\footnote{\citet{monderer_kimplementation_2004} claimed NP-hardness of \pHG, but the proof was erroneous~\cite{eidenbenz_cost_2011}.}
We note that hardness for mixed-strategies or positive budgets is less surprising because there are a priori more possibilities for encoding combinatorial structure into the solutions.
Instead, our reduction shows that the difficulty lies already and mainly in selecting, for each undesired strategy~$x$, a desired strategy that dominates~$x$.

We then study a variant of \pHG\ that was supposedly more tractable \cite{monderer_kimplementation_2004}, called \pEHG:
In addition to requiring undominated strategies to be desired, we require that all desired strategies are undominated.
We show that also this a priori simpler-looking problem is NP-hard even for two players (\cref{thm:EHG-nphard}); this answers an open question by \citet{eidenbenz_cost_2011}.
Indeed \citet{monderer_kimplementation_2004} gave a polynomial-time algorithm for \pEHG\ which was shown to produce suboptimal results by \citet{eidenbenz_cost_2011}.

The above hardness results do not apply in scenarios in which players have only a small constant number of strategies to choose from.
We hence consider this regime next.
If both the number of players and the number of strategies are small constants, then the only part of the input that may be of unbounded size are the quantities specified in the utility functions.
\citet{eidenbenz_cost_2011} showed that in this case \pEHG\ can be solved efficiently; however, as we observe here there is a flaw in the algorithm.
We simplify the algorithm and repair the flaw for a large though not universal class of games, providing the first nontrivial algorithm for implementing strategies that is formally proven to be correct (\cref{thm:eiden_fix}).

\looseness=-1
As we increase the number of players, the size of the input (the number of utility values we have to specify) scales exponentially in the number of players (and strategies).
A common way to deal with this explosion is to instead consider the relevant special case of graphical games \cite{KearnsLS01}, where the players are situated in a graph and the utility of a player depends only on its neighbors.
We hence study this case next, that is, \pHG\ on graphical games with small constant number of strategies per player.
We show that even the case where each player's utility depends only on three others and each player has only two strategies remains NP-hard~(\cref{thm:dg3_2s}).
As the case where each player has only one strategy is trivial, a promising future direction is to consider the case where each player depends only on two others or tree-structured games.

Finally, before discovering the NP-hardness of \pHG\ we believed that zero-cost implementation could be solved efficiently.
As a tool towards this we characterized strategy sets that can be implemented at cost~0 as a form of stable core of the game or, alternatively, as a form of generalized Nash equilibrium.
We believe that this characterization is of independent interest, in particular because it generalizes the result of \citet{monderer_kimplementation_2004} that states that Nash equilibria can be implemented at cost~0.
Moreover, it captures a fundamental property of self-enforcing sets of strategies, such as morality, which we are not aware of having been formally defined before.

\looseness=-1
\paragraph{Further related work.}
Implementation theory~\cite{maskin_nash_1999,maskin_implementation_2002} generally studies the implementation of social-choice rules with incentives and it is impossible to give an overview over the large body of work here.
\citet{conitzer_complexity_2000} studied the complexity of implementing social-choice rules.
The main difference to \pHG\ is that the payment promises to the players that we may choose from are restricted and given in the input.
Such restrictions give significantly more leeway for designing hardness reductions.
\citet{brill_computing_2015} considered a problem related to \pHG\ in which some of the utility values are missing and we are to complete the missing values, possibly with negative ones.
The goal is to ensure that strategies participating in some Nash equilibrium are desired.
In \pHG\ we have much more freedom in designing our solutions.
\citet{wooldridge_incentive_2013} studied implementation questions for Boolean games, that is, where the strategies of the players correspond to a selection of truth values of some variables intrinsic to the game.
They aimed at implementing Boolean formulas on the variables in some or all Nash equilibria.
Because of the relation to Boolean satisfiability, implementation for Boolean games is situated higher in the polynomial hierarchy.
Zero-cost implementation has been studied for routing games by \citet{MoscibrodaS09}.
They gave bounds on the difference between anarchistic equilibria and those achievable by zero-cost implementation.
Finally, \citet{letchford_computing_2010,DengC17,DengC18} considered the complexity of committing to certain behaviors as a way for one player to change the outcome of a game to his favor.

\section{Preliminaries}\label{sec:prelims}
\appendixsection{sec:prelims}

\ifshort(Full) proofs for results marked by \appsymb\ are deferred to the full version \cite{CKHS2022implementgames}.\fi
Throughout, for $t \in \mathbb{N}$ we use $[t]$ to denote the set~$\{1, 2, \ldots, t\}$.

\looseness=-1
A \emph{game} $G$ is a tuple $(N, \mathcal{X}, \mathcal{U})$ where $N$ is the set of \emph{players}; usually $N = [n]$.
We specify for each player $i \in N$ a set $X_i$ of \emph{strategies} available to~$i$.
Then the set $\mathcal{X}$ equals $X_1 \times X_2 \times \ldots \times X_n$.
We call elements of $\mathcal{X}$ \emph{strategy profiles}.
Finally, $\mathcal{U} = \{U_1, U_2, \ldots, U_n\}$, where each $U_i$ is a function $\mathcal{X} \to \mathbb{R}$, called \emph{utility function} for player~$i$.

As a notational shorthand, for any $i \in N$ we use $\mathcal{X}_{-i}$ to denote $X_1 \times X_2 \times \ldots \times X_{i - 1} \times X_{i + 1} \times \ldots \times X_n$.
We sometimes write the value of a utility function $U_i$ such that the strategy played by player~$i$ comes first in the argument of $U_i$ and the remaining strategies second.
That is, if some strategy profile $x \in \mathcal{X}$ consists of strategy $x_i \in X_i$ of player $i$ and a tuple $x_{-i} \in X_{-i}$ of remaining strategies, we write $U_i(x_i, x_{-i})$ for $U(x)$.
However, if there are only two players, then the first argument~$x$ of $U_i(x, y)$ always refers to the strategy of player~$1$ and the second argument~$y$ always refers to the strategy of player~$2$.

Let $x, y \in X_i$ be two strategies of player~$i$.
We say that \emph{$x$ dominates $y$} if for each $x_{-i} \in \mathcal{X}_{-i}$ we have $U_i(x, x_{-i}) \geq U_i(y, x_{-i})$ and there exists $x_{-i} \in \mathcal{X}_{-i}$ such that $U_i(x, x_{-i}) > U_i(y, x_{-i})$.\footnote{This notion of domination is commonly referred to as weak domination. Alternative notions of domination are also studied, such as strict domination in which we require $U_i(x, x_{-i}) > U_i(y, x_{-i})$ for all $x_{-i} \in \mathcal{X}_{-i}$. In keeping with the literature on implementation~\cite{eidenbenz_cost_2011,deng_complexity_2016} we focus on weak domination; the results are usually transferable.}
We say that $x \in X_i$ is \emph{undominated} if no other strategy of player~$i$ dominates~$x$.
For each player $i \in N$ we denote by $X_i^{\star}$ the set of undominated strategies in $X_i$.
For a game $G$, by $\mathcal{X}^{\star}_{G}$ we denote the set of strategy profiles that consist entirely of undominated strategies, that is, $\mathcal{X}^{\star}_G = X_1^{\star} \times X_2^{\star} \times \ldots \times X_n^{\star}$.
We omit the index $G$ if it is clear from the context.

Let $G = (N, \mathcal{X}, \mathcal{U})$ be a game.
We now define the modified game obtained from $G$ after additional payments are promised to the players.
A \emph{payment promise} to player $i$ in $G$ is a function $\mathcal{X} \to \mathbb{R}$, usually denoted by $V_i$.
A \emph{payment promise} for game $G$ is the Cartesian product of the payment promises of the players: $\vvv \coloneqq V_1 \times \dots V_n$.
The \emph{modified game} $G[\mathcal{V}]$\todoH{I find the notation $G[\mathcal{V}]$ misleading as it usually means something with restriction to $\mathcal{V}$.} obtained from $G$ with a payment promise~$\mathcal{V}$ is the game $(N, \mathcal{X}, [\mathcal{U} + \mathcal{V}])$,
wherein $[\mathcal{U} + \mathcal{V}] \coloneqq \{[U_i + V_i] \mid i \in N\}$
and for $i\in N$ function~$[U_i + V_i]$ is defined as $[U_i + V_i](x) \coloneqq U_i(x) + V_i(x)$ for all~$x \in \mathcal{X}$.
The \emph{cost} $\cost(\mathcal{V})$ of a payment promise $\mathcal{V}$ is $\max_{x \in \mathcal{X}^{\star}_{G[\mathcal{V}]}} \sum_{i \in N} V_i(x) $.

We consider the following decision problem and say that
a payment promise~$\mathcal{V}$ as below \emph{implements} $\mathcal{O}$.
\probdef{\pHG}
{A game $G = (N, \mathcal{X}, \mathcal{U})$, a set of desired strategy profiles $\mathcal{O} \subseteq \mathcal{X}$, and a real $\delta \in \mathbb{R}_{\geq 0}$.}
{Is there a payment promise $\mathcal{V}$ such that $\cost(\mathcal{V}) \leq \delta$ and $ \mathcal{X}^{\star}_{G[\mathcal{V}]} \subseteq \mathcal{O}$?}

The following is a variant of \pHG\ where we want to a given set of strategy profiles to be undominated and 
say that a payment promise~$\mathcal{V}$ as below \emph{implements} $\mathcal{O}$ \emph{exactly}.
\probdef{\pEHG}
{A game $G = (N, \mathcal{X}, \mathcal{U})$, a set of desired strategy profiles $\mathcal{O} \subseteq \mathcal{X}$, and a real $\delta \in \mathbb{R}_{\geq 0}$.}
{Is there a payment promise $\mathcal{V}$ such that $\cost(\mathcal{V}) \leq \delta$ and $\mathcal{X}^{\star}_{G[\mathcal{V}]} = \mathcal{O}$?}

\ifshort We give an example in the full version. \else We give an example in \cref{sec:example}. \fi

We mainly focus on the special case where the strategy profile sets~$\ooo$ are rectangular.
A strategy profile set $\mathcal{Y}$ for a game with player set $N$ is \emph{rectangular} if for each $i \in N$ there is $Y_i \subseteq X_i$ such that $\mathcal{Y} = Y_1 \times Y_2 \times \ldots \times Y_n$.
All our hardness results indeed hold even for rectangular strategy profile sets~$\ooo$.

\paragraph{Graphical games.}
For a more succinct representation we also use the concept of a \emph{graphical game}.
This is a tuple $(G, H)$, where $G = (N, \xxx, \uuu)$ is a game and $H$ an undirected graph with vertex set~$N$ and edge set~$E$. Let us use $\ngb_H(i)$ to denote the neighborhood of a vertex $i \in N$, i.e., the set of all the vertices that are adjacent to $i$.
For every $i \in N$, the utility function $U_i$ and the potential payment promise $V_i$ map from $X_i \times X_{j_1} \times \dots \times X_{j_k}$ to $\mathbb{R}$, where $j_1, \dots, j_k$ is an arbitrary but fixed ordering of $\ngb_H(i)$. In other words, the utility of player~$i$ only depends on its own actions and the actions of its neighbors in $H$.
The \emph{degree} of $(G, H)$ is $\max_{i \in N} |\ngb_{H}(i)|$.

\paragraph{Properties of the domination relation.}
We now describe a few simple properties of the dominance relation, which are useful in our proofs. 

\ifshort\begin{observation}[\appsymb]\else\begin{observation}\fi\label{obs:transitive}
  Domination is transitive.
\end{observation}
\appendixproof{obs:transitive}{
  \begin{proof}
    Let $G = (N, \xxx, \uuu)$ be a game.
    Assume that for player $i \in N$ and three strategies~$x,y,z\in X_i$,
    strategy $x$ dominates $y$ and $y$ dominates $z$.
    Then, for every $x_{-i} \in \xxx_{-i}$, $U_i(x, x_{-i}) \geq U_i(y, x_{-i}) \geq U_i(z, x_{-i})$. We also have that there is an $x'_{-i} \in \xxx_{-i}$ such that $U_i(x, x'_{-i}) > U_i(y, x'_{-i})$. Therefore $U_i(x, x'_{-i}) > U_i(y, x'_{-i}) \geq U_i(z', x_{-i})$, so $x$ dominates~$z$.
  \end{proof}
}

\ifshort\begin{observation}[\appsymb]\else\begin{observation}\fi\label{obs:asymmetric}
  Domination is asymmetric. In other words, if $x$ dominates $y$, then $y$ does not dominate $x$.
\end{observation}
\appendixproof{obs:asymmetric}{
  \begin{proof}
    Let $G = (N, \xxx, \uuu)$ be a game. Assume, towards a contradiction, that for player $i \in N$ there are two strategies~$x,y\in X_i$ such that $x$ dominates $y$ and $y$ dominates~$z$.
    Then, for every $x_{-i} \in \xxx_{-i}$, it holds that $U_i(x, x_{-i}) \geq U_i(y, x_{-i})$ and $U_i(y, x_{-i}) \geq U_i(x, x_{-i})$.
    Thus we have that $U_i(x, x_{-i}) = U_i(y, x_{-i})$ holds for every $x_{-i} \in \xxx_{-i}$.
    However, this means that there exists \emph{no}~$x'_{-i} \in \xxx_{-i}$ such that $U_i(x, x'_{-i}) > U_i(y, x'_{-i})$, a contradiction to $x$ dominating~$y$.
  \end{proof}
}

\ifshort\begin{observation}[\appsymb]\else\begin{observation}\fi\label{obs:undom_dom}
  Every dominated strategy is dominated by some undominated strategy.
\end{observation}
\appendixproof{obs:asymmetric}{
  \begin{proof}
    Assume, towards a contradiction, that for a player $i \in N$ there is a dominated strategy $x \in X_i$ that is not dominated by any undominated strategy. This implies that the strategy $y \in X_i$ that dominates $x$ is in turn dominated by some other strategy $z \in X_i$.
    By transitivity of domination $z$ dominates $x$ and because $x$ is not dominated by any undominated strategy, $z$ is also dominated.
    By asymmetry of domination, $x, y$ and $z$ are all distinct.
    By repeating this argument, we obtain that $x$ is dominated by an infinite number of strategy profiles, a contradiction to the number of strategy profiles being finite.
  \end{proof}
}

This immediately implies the following observation:
\begin{observation}\label{obs:undom_noempty}
  The set of undominated strategies is non-empty.
\end{observation}

\toappendix{
  \subsection{Comment on NP-containment of \pHG}
  Previously it was claimed that a payment promise is a certificate showing that \pHG\ is in NP~\cite{deng_complexity_2016}.
  However, without further argument it is not clear that the payment promise can be encoded with polynomial bits in the input length since it consists of arbitrary reals.
  We do conjecture \pHG\ (and indeed \pEHG) to be contained in NP but since there is currently no published proof, we refrain from claiming NP-completeness.
}

\toappendix{
  \subsection{Example}\label{sec:example}
  In this section we provide an example of \pHG\ to give the reader more intuition of the problem.

  Consider a 2-player game with the set of players $N = \{p_1, p_2\}$. Player $p_1$ has the strategy set $X_1 = \{s_1, s_2, s_3\}$ and $p_2$ has $X_2 = \{t_1, t_2\}$. Let us define the utility function of $p_1$ as:
  \begin{align*}
    U_1(s_1, t_1) = 1, \quad & U_1(s_1, t_2) = 1,\\
    U_1(s_2, t_1) = 2, \quad & U_1(s_2, t_2) = 0,\\
    U_1(s_3, t_1) = 0, \quad & U_1(s_3, t_2) = 1,
  \end{align*}
  and the utility function of $p_2$ as:
  \begin{align*}
    U_2(s_1, t_1) = 1, \quad & U_2(s_1, t_2) = 1,\\
    U_2(s_2, t_1) = 1, \quad & U_2(s_2, t_2) = 1,\\
    U_2(s_3, t_1) = 0, \quad & U_2(s_3, t_2) = 0.
  \end{align*}
  We visualize the utility functions in \cref{tab:exp1}.
  \begin{table}[t!]
    \centering
    \setlength{\extrarowheight}{2pt}
    \begin{tabular}{cc|c|c|}
      & \multicolumn{1}{c}{} & \multicolumn{2}{c}{Player $p_2$}\\
      & \multicolumn{1}{c}{} & \multicolumn{1}{c}{$t_1$}  & \multicolumn{1}{c}{$t_2$} \\\cline{3-4}
      \multirow{3}*{Player $p_1$}  & $s_1$ & $1, 1$ & $1, 1$ \\\cline{3-4}
      & $s_2$ & $2, 1$ & $0, 1$ \\\cline{3-4}
      & $s_3$ & $0, 0$ & $1, 0$ \\\cline{3-4}
    \end{tabular}
    \caption{The utility matrix from the example in \cref{sec:example}. The first entry is the utility of $p_1$ and the second the utility of $p_2$.}\label{tab:exp1}
  \end{table}
  \begin{table}[t!]
    \centering
    \setlength{\extrarowheight}{2pt}
    \begin{tabular}{cc|c|c|}
      & \multicolumn{1}{c}{} & \multicolumn{2}{c}{Player $p_2$}\\
      & \multicolumn{1}{c}{} & \multicolumn{1}{c}{$t_1$}  & \multicolumn{1}{c}{$t_2$} \\\cline{3-4}
      \multirow{3}*{Player $p_1$}  & $s_1$ & $1 {\color{green!80!black} {} + 1}, 1 {\color{red} {} + 0.1}$ & $1, 1$ \\\cline{3-4}
      & $s_2$ & $2, 1$ & $0, 1$ \\\cline{3-4}
      & $s_3$ & $0, 0$ & $1, 0$ \\\cline{3-4}
    \end{tabular}
    \caption{The utility matrix from the example in \cref{sec:example} after adding the payment promises $\vvv$. The utility in red is $V_2$.}\label{tab:exp2}
  \end{table}
  \begin{table}[t!]
    \centering
    \setlength{\extrarowheight}{2pt}
    \begin{tabular}{cc|c|c|}
      & \multicolumn{1}{c}{} & \multicolumn{2}{c}{Player $p_2$}\\
      & \multicolumn{1}{c}{} & \multicolumn{1}{c}{$t_1$}  & \multicolumn{1}{c}{$t_2$} \\\cline{3-4}
      \multirow{3}*{Player $p_1$}  & $s_1$ & $1 {\color{green!80!black} {} + 1}, 1$& $1, 1$ \\\cline{3-4}
      & $s_2$ & $2, 1 {\color{blue} {} + 0.1}$ & $0, 1$ \\\cline{3-4}
      & $s_3$ & $0, 0$ & $1, 0$ \\\cline{3-4}
    \end{tabular}
    \caption{The utility matrix from the example in \cref{sec:example} after adding the payment promises $\vvv$. The utility in blue is $V'_2$.}\label{tab:exp3}
  \end{table}
  We can see that no matter what $p_2$ plays, strategy~$s_1$ is as good as strategy~$s_3$ for player~$p_1$, i.e., the utility that $p_1$ gets by playing~$s_1$ is at least as high as that by playing~$s_3$.
  Additionally, if $p_2$ plays~$t_1$, then the utility that $p_1$ gets by playing $s_1$ rather than $s_3$ is strictly higher. Thus, $s_1$ dominates $s_3$ for $p_1$,
  and so $s_3$ is dominated. Since no other strategy dominates $s_1$, strategy~$s_1$ is undominated.

  The strategy $s_2$ is also undominated, because neither $s_1$ nor $s_3$ dominates it: We can see that $U_1(s_2, t_1) > U_1(s_1, t_1)$, so it is impossible that for every $t_i \in X_2, U_1(s_1, t_i) \geq U_1(s_2, t_i)$, which is one of the requirements for $s_1$ dominating $s_2$.
  The same reasoning precludes $s_3$ from dominating~$s_1$.
  Summarizing, the set of undominated strategies for $p_1$ is $X_1^{\star} = \{s_1, s_2\}$.

  For player~$p_2$, it holds that $U_2(s, t_1) = U_2(s, t_2)$ for every $s \in X_1$.
  This means that neither of them dominates the other one.
  Thus, the set of undominated strategies for player~$p_2$ is $X_2^{\star} = \{t_1, t_2\}$.
  The set of undominated strategy profiles is $\xxx^{\star} = \{s_1, s_2\} \times \{t_1, t_2\}$.

  Assume that we want to implement the strategy profile set $\ooo = O_1 \times O_2$, where $O_1 = \{s_1, s_3\}$ and $O_2 = \{t_1\}$.

  Because $s_2 \notin O_1$ but $s_2$ is undominated, we need to use some strategy in $O_1$ to dominate $s_2$.
  If we promise $p_1$ utility $1$ extra when he plays $s_1$ and $p_2$ plays $t_1$,
  then strategy~$s_1$ dominates $s_2$.
  In other words, we construct our payment promise function for $p_1$ as \[V_1(x, y)  = \begin{cases}1, \text{ if }x = s_1, y = t_1\\ 0, \text{ otherwise.}\end{cases}\]

  \noindent The payment promise for $p_1$ is in {\color{green!80!black} green} in \cref{tab:exp2}.
  As for player~$p_2$, we must make $t_1$ dominate $t_2$.
  If we promise
  \[V_2(x, y)  = \begin{cases}0.1, \text{ if }x = s_1, y = t_1, \\ 0, \text{ otherwise,}\end{cases}\]
  then $t_1$ will dominate $t_2$; see the promise in {\color{red} red} in \cref{tab:exp2}.
  Let $\mathcal{V}$ denote the set of payment promises, i.e., $\vvv = \{V_1, V_2\}$.
  Then, $\mathcal{V}$ implements $\ooo$ because  $ X_1^{\star} = \{s_1\}$ and $X_2^{\star} = \{t_1\}$ and thus $\mathcal{X}^{\star}_{G[\mathcal{V}]}=X_1^{\star} \times X_2^{\star}$ satisfies $\emptyset \neq \mathcal{X}^{\star}_{G[\mathcal{V}]} \subseteq \ooo$.
  Note that, since $X_1^{\star} \neq O_1$, this implementation is not exact.
  The cost of the implementation is $\cost(\vvv) = \max_{x \in \xxx^{\star}_{G[\vvv]}}\sum_{i \in N}V_i(x) = V_1(s_1, t_1) + V_2(s_1, t_1) = 1.1$.

  However, instead of paying player~$p_2$ an extra utility of $0.1$ when he plays $t_1$ and player~$p_1$ plays $s_1$, we can also pay $p_2$ an extra utility of $0.1$ when player~$p_1$ plays $s_2$ and player $p_2$ pays $t_1$.
  That is, %

  \[V'_2(x, y)  = \begin{cases}0.1, \text{ if }x = s_2, y = t_1, \\ 0, \text{ otherwise;}\end{cases}\]
  see the {\color{blue} blue} payment promise in \cref{tab:exp3}.
  
  The payment promise for player~$p_1$ remains the same.
  Let $\mathcal{V}'$ denote the promise payment, i.e., $\vvv' = \{V_1, V'_2\}$.
  Then, one can verify that $\mathcal{V}'$ also implements $\ooo$ with $\xxx^{\star}_{G[\vvv']} = \{s_1\}\times \{t_1\}$. %
  But, now the cost of the implementation is only
  \begin{align*}
    \cost(\vvv') &= \max_{x \in \xxx^{\star}_{G[\vvv']}}\left(V_1(x) + V'_2(x)\right) \\
                &= V_1(s_1, t_1) + V'_2(s_1, t_1) \\
                &= 1 + 0 = 1.
  \end{align*}

} %

\section{\pHG\ is NP-hard for Two Players and Zero Budget }\label{sec:hg-nph-2p-0}
\appendixsection{sec:hg-nph-2p-0}

In this section we prove that \pHG\ is NP-hard even in the very restricted case where we have two players and the budget is zero.

\begin{theorem}\label{thm:2p}
  \pHG\ is NP-hard, even for two players and cost at most $0$.
\end{theorem}%
\noindent
We reduce from the following NP-hard problem \cite{schaefer1978complexity}:

\probdef{\pXTC}
{An integer $\enn$, a collection of elements $\aaa = \{a_0, \dots, a_{3\enn - 1}\}$ and a collection of sets $\ccc = \{C_0, \dots, C_{3\enn - 1}\}$ such that $C_j \subset \aaa$ and $|C_j| = 3$ for every $j \in \{0, \dots, 3\enn - 1\}$ and every element $a_i \in \aaa$ appears in exactly three sets, i.e., $|\{C \in \ccc | a_i \in C\}| = 3$ for every $a_i \in \aaa$.}
{Is there an \emph{exact cover} of $\aaa$,
  i.e., a collection $\sss \subset \ccc$
  s.t.\ $|\sss| = \enn$ and $\aaa=\bigcup_{C \in \sss} C$? %
}

Let $\mathcal{I} = (\enn, \aaa, \ccc)$ be an instance of \pXTC.
We create an instance $\mathcal{I'}$ of  \pHG~with two players $p_1$ and $p_2$. Let $X_1 = X_2 = \aaa \cup \{c^{a_i}_j \mid C_j \in \ccc, a_i \in C_j\}$ and $O_1 = O_2 = \{c^{a_i}_j \mid C_j \in \ccc, a_i \in C_j\}$. 
For each element $a_i \in \aaa$ and each set $C_j \in \ccc$ with $a_i \in C_j$, we define the utilities as 
\begin{align*}
  U_1(a_i, c^{a_i}_j) &= 2  & U_1(a_i, c^{a_n}_j) &= 1 \\
  U_1(c^{a_i}_j, c^{a_i}_j) &= 2 & U_1(c^{a_i}_j, c^{a_n}_j) &= 1 \text{ where } a_n \in C_j \setminus a_i.
\end{align*}

Throughout, we take $i + 1$ and $i - 1$ modulo $3\enn$.
For each $a_i \in \aaa$ and each set $C_p \in \ccc$ such that $a_{i - 1} \in C_p$, we define $U_2(c^{a_{i-1}}_p, a_i) = 1$. For $a_i \in \aaa$ and each set $C_p, C_j \in \ccc$ (not necessarily distinct) such that $a_{i - 1} \in C_p$ and $a_i \in C_j$, define $U_2(c^{a_{i-1}}_p, c^{a_{i}}_j) = 1$.
The undefined utilities are $0$ and the budget $\delta$ is also $0$.
The utilities of $p_{1}$ and $p_2$ are shown in \cref{tab:2p_p1}.

\begin{table*}
  \centering
  \small
  \setlength{\tabcolsep}{3pt}
    \setlength{\extrarowheight}{3pt}
    \begin{tabular}{c|c|c|c>{\scriptsize}cc|c|c>{\scriptsize}cc>{\scriptsize}cc|}
         \multicolumn{1}{c}{}\multirow{2}*{\rotatebox{40}{Player $p_1$}} & \multicolumn{11}{c}{Player $p_2$}\\
      \multicolumn{1}{c}{} & \multicolumn{1}{c}{$c^{a_0}_{0,1}$}  & \multicolumn{1}{c}{$c^{a_0}_{0,2}$} & \multicolumn{1}{c}{$c^{a_0}_{0,3}$} & \multicolumn{1}{c}{\scriptsize\dots} & \multicolumn{1}{c}{$ c^{a_i}_{i,1}$} & \multicolumn{1}{c}{$c^{a_i}_{i,2}$}& \multicolumn{1}{c}{$c^{a_i}_{i,3}$} & \multicolumn{1}{c}{\scriptsize\dots} &\multicolumn{1}{c}{$ c^{a_j}_{r}$} & \multicolumn{1}{c}{\scriptsize\dots} & \multicolumn{1}{c}{$c^{a_{3\enn - 1}}_{3\enn-1,3}$}  \\\cline{2-12}
      $c^{a_0}_{0,1}$ & 2 & & & \dots & & & & \dots & & \dots &  \\\cline{2-12}
      $c^{a_0}_{0,2}$ &  &  2& & \dots & & & & \dots & & \dots &  \\\cline{2-12}
      $c^{a_0}_{0,3}$ & & & 2 & \dots & & & & \dots & & \dots & \\\cline{2-12}
      \scalebox{0.5}{\vdots} & \scalebox{0.5}{\vdots} & \scalebox{0.5}{\vdots} & \scalebox{0.5}{\vdots} & &\scalebox{0.5}{\vdots} &\scalebox{0.5}{\vdots} &\scalebox{0.5}{\vdots} & &\scalebox{0.5}{\vdots} & &\scalebox{0.5}{\vdots} \\
      $c^{a_i}_{i,1}$ & & & &\dots & 2 & & & \dots & 1 &\dots &\\\cline{2-12}
      $c^{a_i}_{i,2}$ &&  & &\dots & & 2 & & \dots & & \dots & \\\cline{2-12}
      $c^{a_i}_{i,3}$ & & & &\dots & & & 2 & \dots & & \dots & \\\cline{2-12}
      \scalebox{0.5}{\vdots} & \scalebox{0.5}{\vdots} & \scalebox{0.5}{\vdots} & \scalebox{0.5}{\vdots} & &\scalebox{0.5}{\vdots} &\scalebox{0.5}{\vdots} &\scalebox{0.5}{\vdots} & &\scalebox{0.5}{\vdots} & &\scalebox{0.5}{\vdots} \\
      $c^{a_j}_{r}$ & & & & \dots & 1 &  & & \dots & 2 & \dots & \\
      \scalebox{0.5}{\vdots} & \scalebox{0.5}{\vdots} & \scalebox{0.5}{\vdots} & \scalebox{0.5}{\vdots} & &\scalebox{0.5}{\vdots} &\scalebox{0.5}{\vdots} &\scalebox{0.5}{\vdots} & &\scalebox{0.5}{\vdots} & &\scalebox{0.5}{\vdots} \\
      $c^{a_{3\enn - 1}}_{3\enn - 1,3}$ & & & &\dots & & & & \dots & &\dots &  2 \\\cline{2-12}
      $a_0$ & 2 & 2 & 2 & \dots & & & & \dots & & \dots & \\\cline{2-12}
      \scalebox{0.5}{\vdots} & \scalebox{0.5}{\vdots} & \scalebox{0.5}{\vdots} & \scalebox{0.5}{\vdots} & &\scalebox{0.5}{\vdots} &\scalebox{0.5}{\vdots} &\scalebox{0.5}{\vdots} & &\scalebox{0.5}{\vdots} & &\scalebox{0.5}{\vdots} \\
      $a_i$ & & & & \dots & 2 & 2 & 2 & \dots & 1 & \dots &  \\
      \scalebox{0.5}{\vdots} & \scalebox{0.5}{\vdots} & \scalebox{0.5}{\vdots} & \scalebox{0.5}{\vdots} & &\scalebox{0.5}{\vdots} &\scalebox{0.5}{\vdots} &\scalebox{0.5}{\vdots} & &\scalebox{0.5}{\vdots} & &\scalebox{0.5}{\vdots} \\
      $a_{3\enn - 1}$ & & & & \dots & & & & \dots & & \dots & 2 \\\cline{2-12}
    \end{tabular}
    \setlength{\extrarowheight}{8pt}
    \setlength{\tabcolsep}{5pt}
    \begin{tabular}{c|cccccc|}
      \multicolumn{1}{c}{}\multirow{2}*{\rotatebox{40}{Player $p_2$}} & \multicolumn{6}{c}{Player $p_1$} \\
      \multicolumn{1}{c}{} & $c^{a_{i - 1}}_{i - 1,1}$ & $c^{a_{i - 1}}_{i - 1,2}$ & $c^{a_{i - 1}}_{i - 1,3}$ & $c^{a_i}_{i,1}$ & $c^{a_i}_{i,2}$ & \multicolumn{1}{c}{$c^{a_i}_{i,3}$} \\\cline{2-7}
        $ c^{a_i}_{i,1}$  & 1 & 1 & 1 & & & \\
       $c^{a_i}_{i,2}$ &1 &1 &1 & & & \\
       $c^{a_i}_{i,3}$ &1 &1 &1 & & & \\
       $ c^{a_{i + 1}}_{i + 1,1}$  & & & & 1& 1& 1\\
       $c^{a_{i + 1}}_{i + 1,2}$ & & & & 1& 1& 1\\
       $c^{a_{i + 1}}_{i + 1,3}$ &  & & & 1& 1& 1\\
       \scalebox{0.5}{\vdots}  & \scalebox{0.5}{\vdots} & \scalebox{0.5}{\vdots} & \scalebox{0.5}{\vdots} & \scalebox{0.5}{\vdots} & \scalebox{0.5}{\vdots} & \scalebox{0.5}{\vdots} \\
       $ a_{i} $ & 1 & 1 & 1 & & & \\
       $ a_{i + 1} $  & & & &1 &1 &1 \\\cline{2-7}
    \end{tabular}
    \caption{Left: The utility matrix of $p_1$ from the proof of \cref{thm:2p}. For each element $a_{\ell} \in \aaa$, let $C_{\ell, 1}, C_{\ell, 2}, C_{\ell, 3}$ denote the sets containing it. We assume that $C_{i,1} = C_r$ for some~$r$ and thus $a_i \in C_r$. Columns corresponding to strategies $a_i$ for $p_2$ are omitted: their values are~0.
      Right: The utility matrix of $p_2$ from the proof of \cref{thm:2p}. }\label{tab:2p_p1}
  \end{table*}

  \smallskip
  Before continuing with the proof, let us explain some intuition.
  Let $a_i \in \aaa$ and denote the sets containing $a_i$ as $C_{i, 1}, C_{i, 2}, C_{i, 3}$.
  The utilities of value 2 for $p_1$ enforce that for every $a_i \in \aaa$ exactly one of the strategies $c^{a_i}_{i,1}, c^{a_i}_{i,2}, c^{a_i}_{i,3}$ can be undominated for $p_2$:
  To dominate $a_i$ for $p_1$ with, say, $c^{a_i}_{i,2}$, we must promise a positive amount for playing $c^{a_i}_{i,2}$ whenever $p_2$ plays $c^{a_i}_{i,3}$ or $c^{a_i}_{i,1}$.
  Thus we must have that neither of those latter strategies is undominated for $p_2$ in order to stay within the budget.

  The utilities of value 1 for $p_1$ enforce consistency, i.e., if $a_s \in \aaa$ is covered by $C_r$, then every $a_i \in C_r$ must also be covered by $C_r$. If $c^{a_s}_r$ is undominated for $p_2$, then if we were to dominate $a_i \in C_r$ with $c^{a_i}_{i, 2}$ where $C_{i,2} \neq C_r$, we would have to promise $c^{a_i}_{i, 2}$ to pay at least 1 when $p_2$ plays  $c^{a_s}_r$.
  But since these are both undominated strategies, we would exceed the budget.

  The utilities for $p_2$ enforce that we cover every element $a_i \in \aaa$.
  Player $p_1$ always tries to match the element $p_2$ is playing, whereas $p_2$ tries to be one ahead.
  This prevents the two players from picking some element $a_n \in \aaa$ and only playing the strategies related to that.

  \smallskip
  Formally, we claim that $\mathcal{I}$ is a positive instance of \pXTC~if and only if $\mathcal{I'}$ is a positive instance of \pHG.

For the forwards direction, assume that $(\enn, \aaa, \ccc)$ admits an exact cover $\sss$. 
For each element $a_i \in \aaa$ and two distinct sets $C_j, C_p \in \ccc$ with $a_i \in C_j \cap C_p$, where $C_j \in \sss$ and for each element $a_n \in C_j$, we define $V_1(c^{a_i}_j, c^{a_n}_p) = \infty$.
For each element $a_i \in \aaa$ and two distinct sets $C_j, C_p \in \ccc$ with $a_i \in C_j \in \sss$, while $a_{i -1} \in C_p$ but $C_p \notin \sss$ , we define $V_2(c^{a_{i-1}}_p, c^{a_i}_j) = \infty$.
To show this is a valid implementation, we will observe that $\xxx^{\star} \subseteq \ooo$ and $\cost(\vvv) = 0$.
\ifshort\begin{claim}[\appsymb]\else\begin{claim}\fi\label{obs:2p1}
Since $\sss$ is an exact cover, for every $a_i \in \aaa$ there is exactly one $C_j \in \ccc$ such that $a_i \in C_j$. For such $a_i, C_j$ and for every $C_l \in \ccc \setminus \{C_j\}$ such that $a_i \in C_l$, we have that $c^{a_i}_j$ dominates both $a_i$ and every $c^{a_i}_l$ for both $p_1$ and $p_2$.
\end{claim}

\begin{proof}[Proof of \cref{obs:2p1}]
Observe that for $p_1$, $a_i$ dominates $c^{a_i}_j$ for every $C_j \in \ccc$ where $a_i \in C_j$ in $G$. Since $\sss$ is an exact cover, no $C_l \in \ccc \setminus \{C_j\}$ such that $a_i \in C_l$ is in $\sss$. Therefore payment promise $V_1$ is $0$ when $p_1$ plays $c^{a_n}_l$. Thus $a_i$ also dominates $c^{a_i}_l$ in $G[\vvv]$. Because dominance is transitive (\cref{obs:transitive}), it is enough to show that $c^{a_i}_j$ dominates~$a_i$.

To see that $c^{a_i}_j$ dominates $a_i$ for $p_1$, we show that the payoff for playing $c^{a_i}_j$ is always higher than or equal to that of playing $a_i$.
\begin{description}

\item[Case 1:] \textbf{$p_2$ plays $c^{a_i}_j$.}

Then, $[U_i + V_i](c^{a_i}_j, c^{a_i}_j) = 2 =  [U_i + V_i](a_i, c^{a_i}_j) $.
\item[Case 2:] \textbf{$p_2$ plays $c^{a_i}_l$ for some $C_l \in \ccc \setminus \{C_j\}$ such that $a_i \in C_l$.}

Then, since $a_i$ is covered by $C_j$, we know that $C_l \notin \sss$. Therefore, $[U_1 + V_1](c^{a_i}_j, c^{a_i}_l) = \infty > 2 = [U_1 + V_1](a_i, c^{a_i}_l)$.
\item[Case 3:] \textbf{$p_2$ plays $c^{a_n}_j$ for some $a_n \in C_j \setminus \{a_i\}$.}

Then, $[U_1 + V_1](c^{a_i}_j, c^{a_n}_j) = 1 = [U_1 + V_1](a_i, c^{a_n}_j) $.
\item[Case 4:] \textbf{$p_2$ plays $c^{a_n}_l$ for some $ a_n \in \aaa \setminus \{a_i\}, C_l \in \ccc \setminus \{C_j\}$ where $a_i, a_n \in C_l$}.

Then, since element~$a_i$ is covered by $C_j$ we know that $C_l \notin \sss$. Therefore $[U_1 + V_1](c^{a_i}_j, c^{a_n}_l) = \infty > 1 = [U_1 + V_1](a_i, c^{a_n}_l)$.
\item[Case 5:] \textbf{$p_2$ plays $c^{a_n}_l$ for some $a_n \in \aaa \setminus \{a_i\}, C_{l} \in \ccc \setminus \{C_j\}$, where $a_i \notin C_{l}, a_n \in C_{l}$.}

Then, $[U_1 + V_1](c^{a_i}_j, (a_n, C_l)) = 0 = [U_1 + V_1](a_i, (a_n, C_l))$.
\item[Case 6:] \textbf{$p_2$ plays $a_n \in \aaa$ ($a_i, a_n$ not necessarily distinct).}

Then, $[U_1 + V_1](c^{a_i}_j, a_n) = 0 = [U_1 + V_1](a_i, a_n) $
\end{description}
In all cases, the new utility of $c^{a_i}_j$ is higher than or equal to the utility of $a_i$, and in the Cases 2 and 4 the utility is strictly higher. Thus $c^{a_i}_j$ dominates $a_i$ for $p_1$.
\ifshort The analogous observation for $p_2$ is proved in the full version of the paper.\fi
\appendixproof{obs:2p1}{

Let $a_i \in \aaa$. Observe that in $G$, $p_2$ obtains the same utility for every strategy relating to $a_i$. Formally, for every $s, t \in A^*_i \coloneqq \{a_i\} \cup \{c^{a_i}_j \mid C_j \in \ccc, a_i \in C_j\}$, $U_2(x, s) = U_2(x,t)$ where $x \in X_1$ is an arbitrary strategy. The payment promise $V_2$ is non-zero only for $c^{a_i}_j$, where $C_j \in \sss$. Since $\sss$ is an exact cover, no $C_l \in \ccc \setminus \{C_j\}$ such that $a_i \in C_l$ is in $\sss$. Therefore payment promise $V_2$ is $0$ for every strategy in $a^*_i \in A^*_i \setminus \{c^{a_i}_j\}$. It follows that for every $x, y \in A^*_i \setminus \{c^{a_i}_j\}$, $U_2(x, s) = U_2(x,t)$ for every $x \in X_1$.

Thus, to show case for $p_2$, it suffices to show that $c^{a_i}_j$ dominates an arbitrary $a^*_i \in A^*_i \setminus \{c^{a_i}_j\}$.

\begin{description}
  \item[Case 1:] \textbf{$p_1$ plays $c^{a_{i - 1}}_p$ for some $C_p \in \sss$ such that $a_{i -1} \in C_p$.}
 Then, $[U_2 + V_2](c^{a_{i - 1}}_p, c^{a_i}_j) = 1 = [U_2 + V_2](c^{a_{i - 1}}_p, a^*_i)$
\item[Case 2:] \textbf{$p_1$ plays $c^{a_{i - 1}}_p$ for some $C_p \in \ccc \setminus \sss$ such that $a_{i -1} \in C_p$.}
Then, 
$[U_2 + V_2](c^{a_{i - 1}}_p, c^{a_i}_j) = \infty > 1 = [U_2 + V_2](c^{a_{i - 1}}_p, a^*_i)$
\item[Case 3:] \textbf{$p_1$ plays $c^{a_n}_p$ for some $a_n \in C_p \setminus \{a_{i-1}\} \in \ccc$.}
If $C_p \notin \sss$, $V_2(c^{a_n}_p, c^{a_i}_j) = \infty$, otherwise $V_2(c^{a_n}_p, c^{a_i}_j) = 0$.

In both cases, $[U_2 + V_2](c^{a_n}_p, c^{a_i}_j) \geq 0 = [U_2 + V_2](c^{a_n}_p, a^*_i)$
\item[Case 4:] \textbf{$p_1$ plays some $a_j \in \aaa$.}

Then, $[U_2 + V_2](a_j, c^{a_i}_j) = 0 = [U_2 + V_2](a_j, a^*_i)$.
\end{description}
In all above cases, the new utility of $c^{a_i}_j$ is higher than the utility of $a^*_i$, and in case 2 the utility of $c^{a_i}_j$ is strictly higher. Thus $c^{a_i}_j$ dominates every $a^*_i \in A^*_i$. This concludes the proof.
}%
\end{proof}

Since the set of undominated strategies is non-empty (\cref{obs:undom_noempty}), \cref{obs:2p1} implies that for every $i \in [2]$, $X^{\star}_i \subseteq \{c^{a_i}_j \mid C_j \in \sss, a_i \in C_j\}$ as all other strategies are dominated. Therefore $\xxx^{\star} \subseteq \{c^{a_i}_j \mid C_j \in \sss, a_i \in C_j\} \times  \{c^{a_i}_j \mid C_j \in \sss, a_i \in C_j\} \subseteq \ooo$, as required. For every $j \in [2]$, we have that $V_j(s_1, s_2) > 0$ only when $s_1$ or $s_2$ is in $\{c^{a_i}_j \mid C_j \notin \sss, a_i \in C_j\}$.
Since none of these strategies is in $X^{\star}_{j'}$, we have that $\cost(\vvv) = \max_{x \in \xxx^{\star}_{G[\vvv]}}\sum_{i \in [2]}V_i(x) = 0$, as required.\\

For the backwards direction, assume that we have a payment promise $\vvv$ such that $\cost(\vvv) = 0$ and $\xxx^{\star} \subseteq \ooo$ in the modified game $G[\vvv]$.

\begin{claim}\label{cla:2p1} Let $C_j \in \ccc, a_i \in \aaa$. If $c^{a_i}_j \in X^{\star}_2$, then in $G[\vvv]$
\begin{compactenum}[(i)]
\item $c^{a_i}_j \in X^{\star}_1$,\label{l:2p1:0}
\item $c^{a_i}_j$ dominates $a_i$ for $p_1$ \label{l:2p1:1},
\item $c^{a_i}_l \notin X^{\star}_2$ where $ C_l \in \ccc \setminus \{C_j\}$, $a_i \in C_l$. \label{l:2p1:2}
\end{compactenum}
\end{claim}

\begin{proof}[Proof of \cref{cla:2p1}]
We first show \eqref{l:2p1:0}.
Since $\vvv$ implements~$\ooo$, by \cref{obs:undom_dom} there is an undominated strategy that dominates $a_i$ for $p_1$.
For a strategy $s \in X^{\star}_1$ to dominate $a_i$ for $p_1$ we need that $V_1(s, c^{a_i}_j) \geq U_1(a_i, c^{a_i}_j)  - U_1(s, c^{a_i}_j)= 2 - U_1(s, c^{a_i}_j)$. Because $s$ is undominated, $(s, c^{a_i}_j) \in \xxx^{\star}_{G[\vvv]}$. By the definition of \cost, we have that $0  \geq \cost(\vvv) = \max_{x \in \xxx^{\star}_{G[\vvv]}} \sum_{i \in [2]}V_i(x) \geq V_1(s, c^{a_i}_j)$. By combining these, we obtain that $U_1(s, c^{a_i}_j) \geq 2$. The only strategy for $p_1$ that satisfies this condition is $c^{a_i}_j$. Since $c^{a_i}_j$ dominates $a_i$, \eqref{l:2p1:1} follows directly.

To prove \eqref{l:2p1:2}, assume that both $c^{a_i}_j$ and $c^{a_i}_l$ are undominated for $p_2$.
By \eqref{l:2p1:0} and \eqref{l:2p1:1}, strategy $c^{a_i}_j$ dominates $a_i$ for $p_1$ and $c^{a_i}_j$ is undominated. Thus $[U_1 + V_1](c^{a_i}_j, c^{a_i}_l) \geq U_1(a_i, c^{a_i}_l) = 2$.
From $U_1(c^{a_i}_j, c^{a_i}_l) = 0$ it follows that $ V_1(c^{a_i}_j, c^{a_i}_l) \geq 2$.
But since $c^{a_i}_j$ is undominated for $p_1$ and $c^{a_i}_l$ for $p_2$, $(c^{a_i}_j,c^{a_i}_l) \in \xxx^{\star}$ and thus $cost(\vvv) = \max_{x \in \mathcal{X}^{\star}} \sum_{i \in [2]} V_i(x) \geq V_1(c^{a_i}_j, c^{a_i}_l) = 2$ which a contradiction to $\cost(V) = 0$.
\end{proof}

\begin{claim}\label{cla:2p2}Let $C_j \in \ccc, a_i \in \aaa$. If $c^{a_i}_j \in X^{\star}_1$, then $c^{a_{i + 1}}_l \in X^{\star}_2$ for some $C_l \in \ccc$ such that $a_{i+1} \in C_l$. \end{claim}

\begin{proof}[Proof of \cref{cla:2p2}]
Since $\vvv$ is a valid solution, by \cref{obs:undom_dom}, there is an undominated strategy $s \in X^{\star}_2$ that dominates $a_{i+1}$ for $p_2$. For $s$ to dominate $a_{i + 1}$ we need that $[U_2 + V_2](c^{a_i}_j, s) \geq U_2(c^{a_i}_j, a_{i + 1}) = 1$.

Because $s$ is undominated, $(c^{a_i}_j, s) \in \xxx^{\star}_{G[\vvv]}$. Therefore $0  \geq \cost(\vvv) = \max_{x \in \xxx^{\star}_{G[\vvv]}} \sum_{i \in [2]}V_i(x) \geq V_2(c^{a_i}_j, s)$. From this follows that $U_2(c^{a_i}_j, s) \geq 1$. The only strategies for $p_2$ that satisfy this condition are $c^{a_{i + 1}}_l$ where $C_l \in \ccc$ such that $ a_{i + 1} \in C_l$.
\end{proof}

\looseness=-1
We know from the definition of valid implementation that $X^{\star}_2 \neq \emptyset$. Therefore there is some $C_j \in \ccc, a_i \in C_j$ such that $c^{a_i}_j \in X^{\star}_2$. By \cref{cla:2p1}\eqref{l:2p1:0} $c^{a_i}_j \in X^{\star}_1$. By \cref{cla:2p2} $c^{a_{i + 1}}_p \in X^{\star}_2$ for some $C_p \in \ccc$ such that $a_{i+1} \in C_p$. By repeating this argumentation $3\enn$ times, we obtain that for every $a_{i'} \in \aaa$, there is $C_{j'} \in \ccc$ such that $a_{i'} \in C_{j'}$, $c^{a_{i'}}_{j'} \in X^{\star}_2$.

This shows that $\sss \coloneqq \{C_j \mid C_j \in \ccc, \exists~a_i \in \aaa, \text{ s.t. } c^{a_i}_j \in  X^{\star}_2 \}$ covers $\aaa$.
To show $\sss$ is an exact cover, we must show that if $C_j \in \sss$ and $a_i \in C_j$ then for every $C_l \in \ccc \setminus \{C_j\}$ such that $a_i \in C_l$, we have that $C_l \notin \sss$.

Assume, towards a contradiction, that some $a_i \in \aaa$ is covered twice, i.e., there are $C_j, C_l \in \sss$ where $a_i \in C_j \cap C_l$.  By \cref{cla:2p1}\eqref{l:2p1:2} we cannot have both $c^{a_i}_j \in X^{\star}_2$ and $c^{a_i}_l \in X^{\star}_2$. Therefore, without loss of generality, assume that $c^{a_i}_j \in X^{\star}_2$ and $c^{a_n}_l \in X^{\star}_2$ for some $a_n \in C_l \setminus \{a_i\}$. Then by \cref{cla:2p1}\eqref{l:2p1:1} $c^{a_i}_j$ dominates $a_i$ for $p_1$. Moreover, $[U_1 + V_1](c^{a_i}_j, c^{a_n}_l) \geq U_1(a_i, c^{a_n}_l) = 1$ because $a_i, a_j \in C_l$.

\looseness=-1
Because $U_1(c^{a_i}_j, c^{a_n}_l) = 0$, we have $V_1(c^{a_i}_j, c^{a_n}_l) \geq 1$.
By \cref{cla:2p1}\eqref{l:2p1:0} we have that $c^{a_i}_j \in X^{\star}_1$ and we have assumed that $c^{a_n}_l \in X^{\star}_2$. Thus $\cost(\vvv) \geq V_1(c^{a_i}_j, c^{a_n}_l) \geq 1 > 0$, a contradiction. Therefore each element $a_i \in \aaa$ is covered exactly once, and $\sss$ is an exact cover.
This finishes the proof of \cref{thm:2p}.

\section{\pHG\ is NP-hard for Max.\ Degree Three and Two Strategies}
\label{sec:nph-maxdeg-3}
\appendixsection{sec:nph-maxdeg-3}

In all earlier reductions, the number of strategies per player has been unbounded. In this section we show that in graphical games even bounding the number of strategies and the degree of players together does not help to lower the complexity.\ifshort\begin{theorem}[\appsymb]\else\begin{theorem}\fi\label{thm:dg3_2s}
  \pHG\ is NP-hard, even on graphical games of degree three, where each player has at most two strategies.
\end{theorem}
\begin{proof}%

To show NP-hardness we reduce from \pXTC.
Let $(\enn, \aaa, \ccc)$ be an instance of \pXTC.
We construct an instance $(N, \xxx, \uuu, \ooo, \delta)$ of graphical \pHG\ in the following way: 
Let $N \coloneqq \aaa \cup \ccc$ be the set of players, let $X_i=\{T_i, F_i\}$ be the set of strategies for player $i \in N$. 

Construct the underlying graph $H \coloneqq (N, E)$, where $E \coloneqq \{\{a_i, C_j\} \mid C_j \in \ccc, a_i \in C_j\}$. It is easy to see that $H$ has degree 3. Throughout this proof, for an element $a_i \in \aaa$, let us denote the sets that include it as $C_i^1, C_i^2, C_i^3$ in an arbitrary but fixed order. When defining utility functions, if the utility of the player does not depend on the strategy played by some other player, the strategy of this player is omitted from the function arguments.

For each element $a_i \in \aaa$, we define
\begin{alignat*}{3}
&U_{a_i}(T_{C_i^1}, T_{C_i^2}, T_{C_i^3}, F_{a_i})  &= & U_{a_i}(F_{C_i^1}, F_{C_i^2}, F_{C_i^3}, F_{a_i})\\
={} &U_{a_i}(T_{C_i^1}, T_{C_i^2}, F_{C_i^3}, F_{a_i}) &= & U_{a_i}(T_{C_i^1}, F_{C_i^2}, T_{C_i^3},F_{a_i})\\ 
={} &U_{a_i}(F_{C_i^1}, T_{C_i^2}, T_{C_i^3}, F_{a_i}) &= & 1
\end{alignat*}
The undefined combinations pay out $0$. For every $C_j \in \ccc$, $U_{C_j}$ is 0 for every strategy profile.
 
For every $a_i \in \aaa$ the set of desired outcomes is $O_{a_i} = \{ T_{a_i} \}$. For every $C_j \in \ccc$, $O_{C_j} = X_{C_j}$

It remains to show that there is an exact cover of $\aaa$ if and only if $(G, \ooo, \delta)$ is a positive instance of \pHG.
\ifshort See the full version.\fi
\appendixproof{thm:dg3_2s}{

For the forwards direction, assume that $(\enn, \aaa, \ccc)$ admits an exact cover $\sss$.

For each element $a_i \in \aaa$, we define
\begin{align*}
V_{a_i}(T_{C_i^1}, T_{C_i^2}, T_{C_i^3}, T_{a_i}) &= \infty,\\
V_{a_i}(T_{C_i^1}, T_{C_i^2}, F_{C_i^3}, T_{a_i})  &= \infty,\\
V_{a_i}(T_{C_i^1}, F_{C_i^2}, T_{C_i^3},T_{a_i}) &= \infty, \\
V_{a_i}(F_{C_i^1}, T_{C_i^2}, T_{C_i^3}, T_{a_i}) &= \infty,  \text{ and }\\
V_{a_i}(F_{C_i^1}, T_{C_i^2}, T_{C_i^3}, T_{a_i}) &= \infty.
\end{align*}

For each set $C_j \in \sss$, let $V_{C_j}(T_{C_j}) = \frac{\delta}{3n}$ and for $C_l \in \ccc \setminus \sss$, let $V_{C_l}(F_{C_l}) = \frac{\delta}{3n}$.

It is easy to see that for all $a_i \in \aaa$, $T_{a_i}$ dominates $F_{a_i}$ in $G[\vvv]$ and thus $X_{a_i}^{\star} = \{T_{a_i}\} = O_{a_i}$. For every $C_j \in \sss$, we have $X^{\star}_{C_j} = \{T_{C_j}\}$ and for every $C_l \in \ccc \setminus \sss$ we have $X^{\star}_{C_l} = \{F_{C_l}\}$. Because $O_{C_j} = X_{C_i}$ for every $C_j \in \ccc$, we trivially have that $X_{C_j}^{\star} \subseteq O_{C_j}$.

It remains to show that $\cost(\vvv) \leq \delta$.

Observe that for every $a_i \in \aaa$, $V_{a_i}(x)$ is 0 for every strategy profile $x$ in which exactly one of its set players $C^j_i$ plays~$T_{C^j_i}$.
Because $\sss$ is an exact cover, we have that for every $a_i \in \aaa$ exactly one $C_j \in \{C^1_i, C^2_i, C^3_i\}$ is in $\sss$ and has $T_{C_j}$ as its sole undominated strategy, whereas the others have the corresponding $F$-strategies as their sole undominated strategies. Therefore $V_{a_i}(x) = 0$ for every $x \in \xxx^{\star}$ and we get that  $\max_{x \in \xxx^{\star}}(\sum_{i \in N} V_i(x)) = \max_{x \in \xxx^{\star}}\sum_{C_j \in \ccc} V_{C_j}(x) \leq \sum_{C_j \in \ccc} \frac{\delta}{3n} = \delta$

Since for each player~$i$, the set of undominated strategies for him is non-empty and a subset of $O_i$ and we do not exceed the budget, we have a positive instance of \pHG.\\

Next we prove the backwards direction, namely that if $(G, \ooo, \delta)$ is a positive instance of \pHG, then $(\enn, \aaa, \ccc)$ admits an exact cover.

Assume $(G, \ooo, \delta)$ be a positive instance of \pHG, i.e., there exists a payoff function $\vvv$, s.t. in the game $((N, \xxx, \uuu+\vvv), H)$ the following holds:
\begin{align*}
    \emptyset \neq \xxx_i^* \subseteq O_i \quad & \forall i \in N, \text{ and }\\
    \max_{x \in \xxx^{\star}}\sum_{i \in N} V_i(x) \leq \delta &
\end{align*}

We prove this with the following two claims:
\begin{claim}\label{cla1_deg3} For every $a_i \in \aaa, x \in \xxx^{\star}_{G[\vvv]}$ exactly one $C_j \in \{C^1_i, C^2_i, C^3_i\}$ plays $T_{C_j}$ in $x$.
\end{claim}

\begin{proof}[Proof of \cref{cla1_deg3}]
  For $a_i \in \aaa$, strategy $T_{a_i}$ dominates $F_{a_i}$ in $G[\vvv]$.
  Every strategy profile $x_{-a_i} \in \xxx_{-a_i}$ that does not have exactly one $C_j \in \{C^1_i, C^2_i, C^3_i\}$ playing $T_{C_j}$, has that $U_{a_i}(x_{-a_i}, F_{a_i}) - U_{a_i}(x_{-a_i}, T_{a_i}) = 1$. Thus $V_{a_i}(x_{-a_i}, T_{a_i}) \geq 1 > \delta$ and therefore $(x_{-a_i}, T_{a_i}) \notin \xxx^{\star}$.
\end{proof}

\begin{claim}\label{cla2_deg3}For all $C_j \in \ccc$ either $X^{\star}_{C_j} = \{T_{C_j}\}$ or $X^{\star}_{C_j} = \{F_{C_j}\}$.
\end{claim}
\begin{proof}[Proof of \cref{cla2_deg3}]
Assume towards a contradiction that there is $ C_j \in \ccc \textrm{ s.t. } \xxx^{\star}_{C_j} = \{T_{C_j},F_{C_j} \}$.
Let $a_i \in C_j$.
Assume $x \in \xxx^{\star}$ is an undominated strategy profile where $C_j$ plays $T_{C_j}$.
Because $T_{C_j}$ is undominated and every player has at least one undominated strategy, such an $x$ exists.
By \cref{cla1_deg3}, we know that of the players $\{C^1_i, C^2_i, C_j\}$ (we assume, without loss of generality, that $C^3_i = C_j$) exactly one plays its $T$-strategy, so $C^1_i$ and $ C^2_i$ play their $F$-strategies.
Thus $F_{C^1_i} \in X^{\star}_{C^1_i}$ and $F_{C^2_i} \in X^{\star}_{C^2_i}$.
Since $C_j$ has both $T_{C_j}$ and $F_{C_j}$ in its set of undominated strategies, there is $x' \in \xxx^{\star}$ where $\{C^1_i, C^2_i, C_j\}$ play $\{F_{C^1_i}, F_{C^2_i}, F_{C_j} \}$, a contradiction to \cref{cla1_deg3}.
\end{proof}

We create a set cover $\sss = \{C_j \in \ccc \mid X^{\star}_{C_j} = \{T_{C_j}\}\}$. By Claims~\ref{cla1_deg3}--\ref{cla2_deg3}, for every $a_i \in \aaa$ there is exactly one $C_j \in \sss$ such that $a_i \in C_j$. Thus, $\sss$ is an exact cover of~$\aaa$.
}%
\end{proof}

\section{\pEHG\ is NP-hard}
\label{sec:exact-nph}
\appendixsection{sec:exact-nph}

The complexity of \pEHG\ has so far been open. In this section we show that it is NP-hard even when we have only two identical players.

\ifshort\begin{theorem}[\appsymb]\else\begin{theorem}\fi\label{thm:EHG-nphard}
  \pEHG\ is NP-hard, even for two players and rectangular desired strategy-profile sets.
\end{theorem}
\begin{proof}%
  \looseness=-1
  We give a reduction from the NP-hard \pTC\ problem~\cite{GJ79} in which are a graph~$H$ and
  need to decide whether $H$ can be \emph{properly colored} with three colors.
  That is, whether we can assign each vertex exactly one color such that no two adjacent vertices receive the same color.
  
  Given an instance $H$ of \pTC\ we proceed as follows to construct an instance of \pEHG\ that consists of a game $G = (N, \mathcal{X}, \mathcal{U})$, a rectangular strategy profile set $\mathcal{O}$, and the real budget $\delta = 1$.
  There are two players, that is, $N = \{1, 2\}$.
  The sets of strategies of the two players are identical, that is, $\mathcal{X} = X_1 \times X_2$ and $X_1 = X_2$.
  Thus, we only describe $X_1$.
  For each vertex in $V(H)$ there is a corresponding strategy in $X_1$, that is, $V(H) \subseteq X_1$.
  We call these \emph{vertex strategies}.
  Furthermore, for each combination of a color $c \in [3]$ and a vertex $v \in V(H)$ there is a strategy $(v, c) \in X_1$.
  We call these \emph{color-choice strategies}, and we use $C$ to denote the set of color-choice strategies, that is, $C = \{ (v, c) \mid v \in V(H) \wedge c \in [3]\}$.
  Finally, we have a set $D$ of $3n$ \emph{dummy strategies}.
  Overall, $X_1 = X_2 = V(H) \cup C \cup D$.

  The strategies to implement are exactly the color-choice strategies, that is, $O_1 = O_2 = C$ and $\mathcal{O} = O_1 \times O_2$.
  Intuitively, the strategy $(v, c)$ in $O_1$ that dominates strategy $v$ after promising payments shall correspond to choosing color $c$ for vertex $v$.

  The utility functions are symmetric, that is, for each $x \in X_1 = X_2$ and $y \in X_2 = X_1$ we have $U_1(x, y) = U_2(y, x)$.
  Thus, we only describe $U_1$ explicitly.
  Moreover, we only give the non-zero values of $U_1$, all values not explicitly mentioned are 0.
  First, for each pair $(v, c_1), (v, c_2)$ of color-choice strategies that correspond to the same vertex $v \in V(H)$ we put \todoS{Hua recommended me not to use tuples as strategies, but rather something like $v^{c_1}$}

$    U_1((v, c_1), (v, c_2))  =
                              \begin{cases}
                                3,  c_1 = c_2 \\
                                2, &c_1 \neq c_2.
                              \end{cases}$
  Second, for each pair $(u, c_1), (v, c_2)$ of color-choice strategies that correspond to two adjacent vertices $u, v \in V(H)$ we put
    $U_1((u, c_1), (v, c_2))  = U_1((v, c_1), (u, c_2)) =
                              \begin{cases}
                                1,  c_1 = c_2 \\
                                2,  c_1 \neq c_2.
                              \end{cases}$
  Third, for each vertex strategy $v$ and each color-choice strategy $(v, c)$ corresponding to $v$ we put $U_1(v, (v, c)) = 3$.
  Fourth, for each vertex strategy $u$ and each color-choice strategy $(u, c)$ corresponding to a neighbor $u \in \ngb_H(v)$ of $v$ we put $U_1(v, (u, c)) = 2$.
  Finally, for each color-choice strategy $x \in X_1$ we pick a distinct dummy strategy $y \in X_2$ \todoS{$X_2 \to D_2$?} and put $U_1(x, y) = 1$. \todoS{I would create sets for vertex and color selector strategies and use those instead of generic $X_1$.}
  This concludes the description of the \pEHG\ instance $\mathcal{I} = (G, \mathcal{O}, \delta)$ where $G = (N, \mathcal{X}, \mathcal{U})$ and $\delta = 1$.

  Intuitively, the players are highly incentivized to play vertex strategies because of the values $U_1(v, (v, c)) = 3$.
  In order to dominate a vertex strategy~$v$, we need to pick a color-choice strategy corresponding to that vertex $v$ because in order to make a different strategy dominate~$v$ we would exceed the budget of $1$.\todoS{Somewhat clumsy sentence, perhaps break into two?}
  The values $U_1((v, c_1), (v, c_2))$ will enforce that both players select the same color for each vertex.
  Afterwards, if two adjacent vertices $u, v$ would receive the same color $c$ then the values $U_1((u, c), (v, c)) = 1 = U_2((u, c), (v, c))$ enforce that we would have to pay both players:
  These two color-choice strategies would have to dominate $u$ (for player 1) and $v$ (for player 2), respectively, and we have $U_1(u, (v, c)) = 2 = U_2((u, c), v)$.

  \ifshort The proof of the correctness is in the full version.\fi
  \appendixproof{thm:EHG-nphard}{%
  We now show the correctness, that is, $\mathcal{I}$ is a yes-instance if and only if $H$ is 3\nobreakdash-colorable.

  Assume first that $H$ admits a proper 3-coloring $\phi \colon V(H) \to [3]$.
  We define the following payment promises $V_1, V_2$.
  The promises are symmetric, that is, for each $x \in X_1 = X_2$ and $y \in X_2 = X_1$ we have $V_1(x, y) = V_2(y, x)$. 
  Thus we only define $V_1$ explicitly.
  For each vertex $v \in V(H)$ let $c = \phi(v)$ \todoS{Would be easier to remember that $c$ is the color of $v$ if we just had $\phi(v)$ and not $c$. Or maybe $c^v$.}and for each $d \in [3] \setminus \{c\}$ put $V_1((v, c), (v, d)) = 1$.
  Furthermore, for each neighbor $u \in \ngb_H(v)$ put $V_1((v, c), (u, c)) = 1$.
  This concludes the description of $V_1$ (and of $V_2$) and thus the set $\mathcal{V}$ of payment promises.
\todoS{Hua recommended me to create a payoff matrix in addition to the formal explanations of the payoffs as an aid to the reader. I think it would be helpful here too.}
  We claim that $\mathcal{O}$ is implemented in $G[\mathcal{V}]$, that is, $\mathcal{X}^{\star}_{G[\mathcal{V}]} = \mathcal{O}$. 
  
  First, we show that for each player $i \in N$ all strategies outside of $O_i$ are dominated by a strategy inside $O_i$.
  By symmetry, it suffices to consider player~1.
  Observe that each dummy strategy $y$ is dominated since player~1 never obtains a nonzero payoff for playing~$y$.
  Consider a vertex strategy $v \in X_1$.
  We claim that $v$ is dominated by $(v, \phi(v)) \in X_1$.
  Indeed, the only non-zero payoff obtained by player~1 for $v$ under $[U_1 + V_1]$ is when player~2 plays a color-choice strategy.
  For each $c \in [3]$ we have $[U_1 + V_1](v, (v, c)) = 3 = [U_1 + V_1]((v, c), (v, c))$, and thus if the color-choice strategy corresponds to $v$ then $(v, c)$ is at least as good as $v$ for player~1.
  Furthermore, for each neighbor $u \in \ngb_H(v)$ and each $c \in [3]$ we have $[U_1 + V_1](v, (u, c)) = 2 = [U_1 + V_1]((v, c), (u, c))$.
  Thus, also if the color-choice strategy of player~2 does not correspond to $v$, then $(v, c)$ is at least as good as $v$ for player~1.
  Finally, there is a distinct dummy strategy $d \in D$ for which $[U_1 + V_1](v, d) = 0 < 1 = [U_1 + V_1]((v, c), d)$.
  Thus, in this case $(v, c)$ is strictly better than $v$, meaning that, indeed, $(v, c)$ dominates $v$.

  To show the claim it remains to show that all strategies in $O_1$ are undominated.
  This, however, is obvious because $O_1$ is exactly the set of color-choice strategies and each such strategy $x$ has a distinct dummy strategy~$d$ such that for all other strategies $z \in X_1$ we have $U_1(x, d) = 1 > 0 = U_1(z, d)$.
  Thus, indeed, $\mathcal{O}$ is implemented in $G[\mathcal{V}]$.
  Thus the claim holds, that is, $\mathcal{O}$ is implemented in $G[\mathcal{V}]$.

  It remains to show that the cost of $\mathcal{V}$ is at most~1.
  Consider a strategy profile $(x, y)$; we show that the sum of payment promises amount to at most~1.
  If $x$ or $y$ is a vertex strategy there are $0$ promises to both players.
  Thus, consider the case where $x$ and $y$ are color-choice strategies, say $x = (u, c)$ and $y = (v, d)$.
  By the definition of $V_1$ and $V_{2}$ we have nonzero promises only if either (1) $u = v$ or (2) $u$ and $v$ are neighbors.
  In case~(1), we have promised payment~1 to player~1 only if $c = \phi(v)$ and $d$ is some other color, and to player~2 only if $d = \phi(v)$ and $c$ is some other color.
  Hence, the promises sum to at most~1.
  In case~(2), we have promised payment~1 to player~1 only if $d = c = \phi(v)$ and to player~2 only if $d = c = \phi(u)$.
  Thus, also in this case the promises sum to at most~1.
  Thus, indeed the cost of $\mathcal{V}$ is at most 1.
  It follows that $\mathcal{I}$ is a yes-instance.

  Assume now that $\mathcal{I}$ is a yes-instance and let $\mathcal{V} = \{V_1, V_2\}$ be a witnessing set of payment promises.
  Consider player~$1$.
  Since $\mathcal{O}$ is implemented in $G[\mathcal{V}]$, each strategy in $V(H)$ is dominated by a color-choice strategy $(u, c)$ in $O_1 = C$.
  Furthermore, since for each $c' \in [3]$ we have $U_1(v, (v, c')) = 3$ the strategy $(u, c)$ dominating $v$ must obtain a payoff of at least~3, that is, $[U_1 + V_1]((u, c), (v, c')) \geq 3$.
  Since the budget is 1, we have for each $c' \in [3]$ that $U_1((u, c), (v, c')) \geq 2$.
  Thus, $u = v$.
  Hence, each vertex strategy $v$ is dominated for player~1 by a color-choice strategy $(v, c)$ corresponding to $v$.
  By an analogous argument, each vertex strategy $v$ is dominated for player~2 by a color-choice strategy $(v, c)$ corresponding to $v$.

  We now show that for each vertex strategy $v$ the two color-choice strategies dominating $v$ for the two players coincide. 
  Assume to get a contradiction that there is a vertex $v$ such that $v$ is dominated for player~1 by $(v, c_1)$ and $v$ is dominated for player~2 by $(v, c_2)$ such that $c_1 \neq c_2$.
  Observe that \[U_1((v, c_1), (v, c_2)) = 2 = U_2((v, c_1), (v, c_2)).\]
  However, since \[U_1(v, (v, c_2)) = 3 = U_2((v, c_1), v) \] we thus have \[V_1((v, c_1), (v, c_2)) \geq 1 \leq V_2((v, c_1), (v, c_2)),\] i.e., the payment promises sum up to $2$ which is a contradiction to our budget being~$1$.
  Thus, each vertex strategy is dominated by the same color-choice strategy for both players.

  For each $v \in V(H)$ define $\phi(v) = c$ where $c$ is such that $v$ is dominated by $(v, c)$ (for both players).
  Observe that $\phi$ is total.
  We claim that $\phi$ is a proper coloring of~$H$.
  For a contradiction, assume that this is not the case, that is, there are two neighbors $u, v$ in $H$ such that $\phi(u) = \phi(v) = c$.
  Recall that $U_1(v, (u, c)) = 2$.
  By definition of $\phi$, for player~$1$ vertex strategy $v$ is dominated by $(v, c)$.
  Furthermore, $U_1((v, c), (u, c)) = 1$ and, since $(v, c)$ dominates~$v$, we have $V_1((v, c), (u, c)) \geq 1$.
  Analogously, $U_2((v, c), u) = 2$ and, for player~$2$, vertex strategy~$u$ is dominated by $(u, c)$.
  Furthermore, $U_2((v, c), (u, c)) = 1$ and thus $V_2((v, c), (u, c)) \geq 1$.
  This is a contradiction to the the fact that the cost for strategy profile $((v, c), (u, c))$ is 1.
  Thus, indeed $\phi$ assigns no two vertices the same color, meaning that it is a proper three-coloring of~$H$.
}%
\end{proof}

\section{Correction to Algorithm for \pEHG%
}
\label{sec:eidenbenz_correct}
\appendixsection{sec:eidenbenz_correct}
\citet{eidenbenz_cost_2011} gave an algorithm which on input of a game $G$ and a desired strategy-profile region $\ooo$, finds the minimum $\delta$ such that $(G, \ooo, \delta)$ is a positive instance of \pEHG.
This is Algorithm 1 in~\cite{eidenbenz_cost_2011}, which for completeness is \ifshort contained in the full version. \else presented in Algorithm \ref{alg:eidenorig}.\fi

\looseness=-1
The algorithm fails to give an exact implementation when a strategy in $O_i$ dominates some other strategy in $O_i$ for a player $i$.
We show an example where it fails and provide a fix for a class of games which we refer to as \equib\ games.

\toappendix{

  \begingroup
    \removelatexerror

    \SetKwFunction{ExactK}{ExactK}
    \begin{algorithm}[H]
      \caption{Algorithm for Exact $k$-Implementation by \citet{eidenbenz_cost_2011}. 
        Note that \cite{eidenbenz_cost_2011} use $k$ for the budget instead of $\delta$.}\label{alg:eidenorig}
      \Input{A game $G$ and a rectangular region $\ooo$ with $\ooo_{-i} \subset \xxx_{-i} \forall i$.}
      \Output{$k^*(\ooo)$.}
      $V_i(x) \coloneqq 0, W_i(x) \coloneqq 0~\forall\ i \in N, o_i \in O_i$\\
      $V_i(o_i, \bar{o}_{-i}) \coloneqq \infty ~\forall\ i \in N, o_i \in O_i, \bar{o}_{-i} \in X_i \setminus O_i$\\
      compute $X^{\star}$\\
      \Return \ExactK{$V, n$}\\
      \Def{\ExactK{$V,i$}}{
        \Input{Payment $V$, current player $i$}
        \Output{$k^*(\ooo)$ for $G(V)$}
        \setcounter{AlgoLine}{0}
        \eIf{$|X_i^{\star}(V) \setminus O_i| > 0$}{
          $s \coloneqq $ any strategy in $X_i^{\star}(V) \setminus O_i;\; k_{best} \coloneqq \infty$\\
          \ForEach{$o_i \in O_i$}{\ForEach{$o_{-i} \in \ooo_{-i}$}{$W_i(o_i, o_{-i}) \coloneqq \max(0, U_i(s, o_{-i}) - (U_i(o_i, o_{-i}) + V_i(o_i, o_{-i})))$\\}
            $k \coloneqq$ \ExactK($\vvv + \mathcal{W}, i$)	\\
            \If{$k < k_{best}$}{$k_{best} \coloneqq k$}
            \ForEach{$o_{-i} \in \ooo_{-i}$}{$W_i(o_i, o_{-i}) \coloneqq 0$}
            \Return $k_{best}$\\
          }
        }
        {
          \eIf{$i > 1$}{\Return \ExactK($V, i - 1$)}{\Return $\max_{o \in \ooo} \sum_{i \in N}V_i(o)$}
        }
      } 
    \end{algorithm}
  \endgroup
}
To see that the algorithm does not always construct a correct payment promise $\vvv$, consider a 2-player instance where player 1 and player 2 both have two strategies $\{s_1, s_2\}$. Let us define the utility functions for both players $i \in [2]$ as
\begin{align*}
U_i(s_1, s_1) &= 2 & U_i(s_2, s_1) &= 1\\
U_i(s_1, s_2) &= 1 & U_i(s_2, s_2) &= 0.
\end{align*}
Let $O_1 = \{s_1, s_2\}$ and $O_2 = \{s_1\}$.

We can see that for both players $s_1$ dominates $s_2$, so $X^{\star}_i = \{s_1\}$ for all $i \in [2]$. Because $|X^{\star}_i  \setminus O_i| = 0$ for all $i \in [2]$, the check on line (1) of Algorithm 1 from~\cite{eidenbenz_cost_2011} is always false and the algorithm returns that $\ooo$ can be implemented exactly with cost 0. However, $V_1$ constructed by the Algorithm is $0$ everywhere, and thus $O_1 \neq X^{\star}_1$, meaning this is not an exact implementation of~$\ooo$.

Towards a correction, if we can find for every pair $(i, o_i)$ of a player $i \in N$ and desired strategy $o_i \in O_i$ at least one distinct undesired strategy profile $x_{-i}^{o_i}$, such that $(o_i, x_{-i}) \notin O$, or in other words, if
\begin{multline}
  \label{eq:equitable}
  |O_i| \leq |(X_1 \times X_2 \times \ldots \times X_{i-1} \times X_{i+1} \times \ldots \times X_n) \setminus \\(O_1 \times O_2 \times \ldots \times O_{i-1} \times O_{i+1} \times \ldots \times O_n)|,
\end{multline}
then we can make certain that no strategy in $O_i$ dominates another strategy in $O_i$.
Let us call a game for which, for every $i \in N$, \cref{eq:equitable} holds %
\emph{\equib}.

\looseness=-1
We start by showing that, if a game is \equib\ then, we can translate every non-exact implementation to an exact implementation, where the cost is bounded by the worst-case payment over $\ooo$ in the initial implementation.

Throughout this section, let $\fff$ denote the Cartesian product of the possible functions from $X_i \setminus O_i$ to $O_i$ for every player, i.e., $\fff = (X_1 \setminus O_1 \to O_1) \times \dots \times (X_n \setminus O_n \to O_n) $.

\SetKwFunction{ComputeV}{ComputeV}
\begin{algorithm}[t!]
\caption{Minimum cost exact implementation}\label{alg:eidennew}
\Input{A game $G = (N, \xxx, \uuu)$ and a rectangular strategy profile region $\ooo = O_1 \times \dots \times O_n$.}
\Output{A payment promise $\vvv$ and $\delta \geq 0$ such that $\vvv$ implements $\ooo$ and $\max_{o \in \ooo} \sum_{i \in N}V_i(o) = \delta$ is smallest possible.}
\ForEach{$i \in N$ \label{algli:initcount1}}
    {
    \ForEach{mapping $F_i \colon X_i \setminus O_i \to O_i $\label{algli:initcount2}}
    	{
          $V^{F_i} \gets $ \ComputeV{$F_i, X_i, O_i$\label{algli:initcount3}}
    	}
    }
$\delta \gets \infty$; $\vvv \gets (\mathit{0}, \dots, \mathit{0})$\;
\ForEach{$F = (F_1, \dots, F_n) \in \fff$ \label{algli:foreachF}}{
	\label{algli:deltadef}$\delta^F \gets \max_{o \in \ooo} \sum_{i \in N}V^{F_i}(o)$\;
	\lIf{$\delta^F < \delta$}
        {$\delta \gets \delta^F$; $\vvv \gets F$}
}\label{algli:foreachFend}
\ForEach{$i \in N$\label{algli:inftypromstart}}
{\ForEach{$o_i \in O_i$}
  {\lForEach{$x_{-i} \in \xxx{-i} \setminus \ooo_{-i}$}
    {$V_i(o_i, x_{-i}) \gets \infty$}}}\label{algli:inftypromend}
\Return $\delta, \vvv$\\
\Def{\ComputeV{$F_i, X_i, O_i$ \label{algli:comv1}}}{
\ForEach{$o_i \in O_i$\label{algli:comv2}}{
	\ForEach{$o_{-i} \in \ooo_{-i}$\label{algli:comv3}}{
			\leIf{ $F_i^{-1}(o_i) \neq \emptyset$}{\\$V_i(o_i, o_{-i}) \gets \max\{0, $\\\hfill$\max_{x_i \in F_i^{-1}(o_i)}U_i(x_i, o_{-i}) - U_i(o_i, o_{-i})\}$\;\label{algli:comv4}}{$V_i(o_i, o_{-i}) \gets 0$}
		}
	}
}
\Return $V_i$
\end{algorithm}

\ifshort\begin{theorem}[\appsymb]\else\begin{theorem}\fi\label{theorem:false_promises}
  Let $G = (N, \xxx, \uuu)$ be an \equib~game and $\ooo \subset \xxx$ a rectangular strategy profile region. Let $\vvv$ implement $\ooo$ (not necessarily exactly) and let $\delta = \max_{o \in \ooo}\sum_{i \in N}$. Then the payment promise $\vvv^*$ defined below implements $\ooo$ exactly and with $\cost(\vvv^*) = \delta$.
  Let $M = U_{max} + \delta + 1$ with $U_{max} = \max_{i \in N, x\in X}U_i(x)$, i.e., $U_{max}$ is the maximum amount any player may receive in utility.
  For every $i \in N, o_i \in O_i$, pick a distinct strategy profile $x_{-i}^{o_i}$ such that $(o_i, x_{-i}) \notin \ooo$.
  Since $G$ is \equib, we can choose a distinct profile for every $i \in N, o_i \in O_i$.

  To define $\vvv^{*}$, for every $i \in N$, $o_i \in O_i$ put
\begin{equation}
  V_i^*(x_i, x_{-i}) =\left\{
  \begin{array}{@{}ll@{}}
    V_i(o_i, x_{-i}), & \text{if } x_{-i} \in \ooo_{-i} \\
    M + 1 - U_i(x_i, x_{-i}), & \text{if } x_{-i} = x_{-i} ^ {x_i} \\
    M- U_i(x_i, x_{-i}), & \text{otherwise.}
  \end{array}\right.
\end{equation}
If $x_i \in X_i \setminus O_i$, then $V^*(x_i, x_{-i}) = 0$ for every $x_{-i} \in \xxx_{-i}$.
\end{theorem}
\noindent\looseness=-1
The idea behind the payments is to enforce that every desired strategy has one undesired strategy profile where they are the best possible option. This prevents any other strategy from dominating it.

\appendixproof{theorem:false_promises}{
Suppose $\vvv$ implements $\ooo$, i.e., $\xxx^{\star}_{G[\vvv]} \subseteq \ooo$.

Assume, towards a contradiction that $\ooo \neq \xxx_{G[\vvv^*]}^{\star}$.

First assume that $\ooo \supset \xxx_{G[\vvv^*]}^{\star}$. Therefore there is a player $i \in N$ that has a desired strategy $o_i \in O_i$, which is dominated by some strategy $x_i \in X_i$. Note that $x_i$ may be in $O_i$. Since $x_i$ dominates $o_i$, we know that for every strategy profile $x_{-i} \in \xxx_{-i}$ we have that $[U_i + V_i^*](x_i, x_{-i}) \geq [U_i + V_i^*](o_i, x_{-i})$. We proceed in two cases:

\begin{description}
\item[Case 1:] $x_i \in O_i$. We have that $[U_i + V_i^*](x_i, x_{-i}^{o_i}) = M < M + 1 = [U_i + V_i^*](o_i, x_{-i}^{o_i})$, so $x_i$ cannot dominate $o_i$.
\item[Case 2:] $x_i \in X_i \setminus O_i$. For every $x_{-i} \in \xxx_{-i} \setminus \ooo_{-i}$, $V_i^*(x_i, x_{-i}) = 0$ and $V_i^*(o_i, x_{-i}) \geq M - U_i(o_i, x_{-i})$. Therefore $[U_i + V_i^*](x_i, x_{-i}) = U(x_i, x_{-i}) < M \leq [U_i + V_i^*](o_i, x_{-i})$, so $x_i$ cannot dominate $o_i$.
\end{description}

Thus we have that $\ooo \subset \xxx_{G[\vvv^*]}^{\star}$.
There thus is a player $i \in N$ strategy $x_i \in X_i \setminus O_i$ such that $x_i$ is undominated under $G[\vvv^*]$.
We know that $x_i$ is dominated by some desired strategy $o_i \in X^{\star}_i \subseteq O_i$ under $G[\vvv]$, but is not longer dominated by $o_i$ in $G[\vvv^*]$. That means there exists a $x_{-i} \in \xxx_{-i}$ such that $[U_i + V_i^*](x_i, x_{-i}) > [U_i + V_i^*](o_i, x_{-i})$, but $[U_i + V_i](x_i, x_{-i}) \leq [U_i + V_i](o_i, x_{-i})$. We proceed in two cases.
\begin{description}
\item[Case 1:] $x_{-i} \in \ooo_{-i}$. In this case $V_i^*(x_i, x_{-i}) = 0$ and $V_i^*(o_i, x_{-i}) = V_i(o_i, x_{-i})$. By the definition of $V^*_i$ and the assumption that $o_i$ dominates $x_i$ in $G[\vvv]$ we have that $[U_i + V_i^*](o_i, x_{-i}) = [U_i + V_i](o_i, x_{-i}) \geq [U_i + V_i](x_i, x_{-i}) \geq U_i(x_i, x_{-i}) = [U_i + V^*_i](x_i, x_{-i})$. This is a contradiction to the assumption $[U_i + V_i^*](x_i, x_{-i}) > [U_i + V_i^*](o_i, x_{-i})$.
\item[Case 2:] $x_{-i} \in \xxx_{-i} \setminus \ooo_{-i}$. We have that $[U_i + V_i^*](o_i, x_{-i}) \geq M$ and $V_i^*(x_i, x_{-i}) = 0$, therefore $[U_i + V_i^*](o_i, x_{-i}) \geq M > U_i(x_i, x_{-i}) = [U_i + V_i^*](x_i, x_{-i})$, a contradiction to the assumption $[U_i + V_i^*](x_i, x_{-i}) > [U_i + V_i^*](o_i, x_{-i})$.
\end{description}

Therefore $\xxx_{G[\vvv^*]}^{\star} = \ooo$.

Last we show that $\cost(\vvv) = \delta$. We have shown that $\xxx_{G[\vvv^*]}^{\star} = \ooo$ and constructed $\vvv^*$ so that $V_i^*(o) = V_i(o)$ for every $o \in \ooo, i \in N$. Therefore $\max_{o \in \ooo}\sum_{i \in N}V^*_i(o) = \max_{o \in \ooo}\sum_{i \in N}V_i(o) = \delta$, as required.
}
Next we show that Algorithm~1 by \citet{eidenbenz_cost_2011} identifies a payment promise $\vvv$ which minimizes $ \max_{o \in \ooo}\sum_{i \in N}V_i(o)$. To make the analysis of the algorithm easier, we have made a simplified version of their algorithm, which is shown in Algorithm~\ref{alg:eidennew}.
\toappendix{

  \subsection{Correctness of Our Simplified Algorithm}
We start by showing some properties of the algorithm.

\begin{observation}\label{obs:undesired_zero}
Let $\vvv$ be a payment promise returned by Algorithm~\ref{alg:eidennew}. Then for every player $i \in N$, undesired strategy $x_i \in X_i \setminus O_i$, strategy profile $x_{-i} \in \xxx_{-i}$, $V_i(x_i, x_{-i}) = 0$.
\end{observation}

\begin{proof}
We initially set $V_i = \mathit{0}$ for every $i \in N$. It is straightforward to observe that for every undesired strategy $x_i \in X_i \setminus O_i$ and a strategy profile $x_{-i} \in \xxx_{-i}$, $V_i(x_i, x_{-i})$ is never modified.
\end{proof}

\begin{claim}\label{cla:alg_returns_valid}
Let $G = (N, \xxx, \uuu)$ be a game and $\ooo$ a rectangular desired strategy profile. Algorithm~\ref{alg:eidennew} returns $\vvv, \delta$ such that $\vvv$ implements $\ooo$ and $\max_{o \in \ooo}\sum_{i \in N}V_i(o) = \delta$.
\end{claim}
\begin{proof}[Proof of \cref{cla:alg_returns_valid}]
Consider $\vvv$ returned from Algorithm~\ref{alg:eidennew}. Lines \eqref{algli:inftypromstart} - \eqref{algli:inftypromend} enforce that for every $i \in N, o_i \in O_i, x_{-i} \in \xxx_{-i} \setminus \ooo_{-i}$, $V_i(o_i, x_{-i}) = \infty$. Lines \eqref{algli:foreachF}-\eqref{algli:foreachFend} enforce that there is $(F_1, \dots, F_n) \in \fff$ such that for every $i \in N, (o_i,o_{-i}) \in \ooo$, $V_i(o_i, o_{-i}) = \max\{0, \max_{x_i \in F_i^{-1}(o_i)}U_i(x_i, o_{-i}) - U_i(o_i, o_{-i})\}$.

Consider an arbitrary player $i \in N$ and an arbitrary undesired strategy $x_i \in X_i \setminus O_i$. 

For each $o_{-i }\in \ooo_{-i}$ we have that $[U_i + V_i](F_i(x_i), o_{-i}) \geq U_i(F_i(x_i), o_{-i}) + U_i(x_i, o_{-i}) - U_i(F_i(x_i), o_{-i}) = U_i(x_i, o_{-i})$. By \cref{obs:undesired_zero}, $[U_i + V_i](x_i, o_{-i}) = U_i(x_i, o_{-i})$. Thus $[U_i + V_i](F_i(x_i), o_{-i}) \geq [U_i + V_i](x_i, o_{-i}) $.

For every $x_{-i} \in \xxx_{-i} \setminus \ooo_{-i}$ we have that $[U_i + V_i](F_i(x_i), x_{-i}) = \infty > U_i(x_i, x_{-i}) = [U_i + V_i](x_i, x_{-i})$.

As we have for every $x_{-i} \in \xxx_{-i}$ that $[U_i + V_i](F_i(x_i), x_{-i}) \geq [U_i + V_i](x_i, x_{-i}) $ and the inequality is strict when $x_{-i} \notin \ooo_{-i}$, $F_i(x_i)$ dominates $x_i$. As this holds for an arbitrary $i, x_i$, $\xxx^{\star}_{G[\vvv]} \subseteq \ooo$.

In line \eqref{algli:deltadef} we set $\delta \coloneqq \max_{o \in \ooo}\sum_{i \in N}V_i(o)$, so this part is trivially true.
\end{proof}

\begin{claim}\label{cla:upper_bound_alg}
Let $F = (F_1, \dots, F_n) \in \fff$ be arbitrary. For every $o \in \ooo$, Algorithm \ref{alg:eidennew} returns a $\delta$  such that $\delta \leq \max\{0,\max_{(o_i, o_{-i}) \in \ooo}\sum_{i \in N}\max_{x_i \in F_i^{-1}(o_i)}U_i(x_i, o_{-i}) - U_i(o_i, o_{-i})\}$.
\end{claim}
\begin{proof}[Proof of \cref{cla:upper_bound_alg}]
Consider line \eqref{algli:foreachF} of Algorithm \ref{alg:eidennew}. For a given $F \in \fff$, there is an iteration of the algorithm where the variable $\mathbf{F}$ coincides with $F$. In that iteration, we have that $\delta^F = \max_{o \in \ooo}\sum_{i \in N}V^{F_i}(o)$, where $V^{F_i}(o_i, o_{-i}) = \max\{0, \max_{x_i \in F_i^{-1}(o_i)}U_i(x_i, o_{-i}) - U_i(o_i, o_{-i})\}$. Thus $\delta^F$ satisfies $\delta^F \leq \max\{0,\max_{(o_i, o_{-i}) \in \ooo}\sum_{i \in N}\max_{x_i \in F_i^{-1}(o_i)}U_i(x_i, o_{-i}) - U_i(o_i, o_{-i})\}$.

The loop on lines \eqref{algli:foreachF} - \eqref{algli:foreachFend} finishes with either variable $\mathbf{\delta} = \delta^F$ or $\mathbf{\delta}  = \delta' \leq \delta^F$. As $\mathbf{\delta}$ is returned with no further modification, this concludes the proof.
\end{proof}

With these properties, we can prove the following of Algorithm 1 from~\cite{eidenbenz_cost_2011}:

\begin{lemma}\label{lem_fix_correct}
Given a game $G = (N, \xxx, \ooo)$ and  a desired rectangular strategy profile region $\ooo \subseteq \xxx$, Algorithm~\ref{alg:eidennew} correctly returns  the smallest $\delta \geq 0$ such that there is a payment promise $\vvv$ for which $\max_{o \in O}\sum_{i \in N}V_i(o) \leq \delta$ and $\xxx^{\star}_{G[\vvv]} \subseteq \ooo$.
\end{lemma}

\begin{proof}[Proof of \cref{lem_fix_correct}]
By \cref{cla:alg_returns_valid} we know that Algorithm~\ref{alg:eidennew} returns a payment promise $\vvv$ that implements $\ooo$ and $\delta = \max_{o \in \ooo}\sum_{i \in N}V_i$.

It remains to show that if there is a payment promise $\vvv$  and $\delta \geq 0$ with $ \max_{o \in \ooo}\sum_{i \in N}V_i(o) \leq \delta$ such that $\vvv$ satisfies $\xxx^{\star}_{G[\vvv]} \subseteq \ooo$, the algorithm returns some $\delta' \leq \delta$.

Assume that $\vvv$ satisfies $\xxx^{\star}_{G[\vvv]} \subseteq \ooo$.
Then, for every $i \in N, x_i \in X_i \setminus O_i$ there is at least one desired strategy $o_i \in O_i$ such that $o_i$ dominates $x_i$ for player $i$.
Let $\hat{F}_i : X_i \to O_i$ be a function that encodes this, i.e., $\hat{F}(x_i)$ is an arbitrary but fixed strategy in $O_i$ that dominates $x_i$ for player $i$.

\begin{claim}\label{cla:lower_bound_alg}
  For each player $i \in N$ and each $(o_i, o_{-i}) \in \ooo$
  \[V_i(o_i, o_{-i}) \geq \quad \smashoperator{\max_{x_i \in X_i \setminus O_i, \hat{F}_i(x_i) = o_i}} U_i(x_i, o_{-i}) - U_i(o_i, o_{-i}).\]
If $\{\hat{F}_i(x_i) \mid x_i \in X_i \setminus O_i\} = \emptyset$, then instead $V_i(o_i, o_{-i}) \geq 0$.
\end{claim}

\begin{proof}[Proof of \cref{cla:lower_bound_alg}]
  Since $\hat{F}_i(x_i)$ dominates $x_i$, we have for all $o_{-i} \in O_{-i}$ that $V_i(\hat{F}_i(x_i), o_{-i}) \geq U_i(x_i, o_{-i}) - U_i(\hat{F}_i(x_i), o_{-i})$.
  By considering this from the point of view of $F(x_i) \coloneqq o_i$, we obtain $V_i(o_i, o_{-i}) \geq \max_{x_i \in X_i \setminus O_i, \hat{F}_i(x_i) = o_i}U_i(x_i, o_{-i}) - U_i(o_i, o_{-i})$.
\end{proof}

Let $\delta'$ be the return value of Algorithm~\ref{alg:eidennew} and $\vvv'$ the payment promise constructed by it. By \cref{cla:lower_bound_alg} and \cref{cla:upper_bound_alg} we obtain that $\delta' = \max_{o \in \ooo}\sum_{i \in N}V'_i(o) \leq  \max_{o \in \ooo}\sum_{i \in N}\max_{x_i \in X_i \setminus O_i, \hat{F}_i(x_i) = o_i}U_i(x_i, o_{-i}) - U_i(\hat{F}_i(x_i), o_{-i}) \leq \max_{o \in \ooo}V_i(o_i, o_{-i}) = \delta$, as required.
\end{proof}
}%
We obtain the following:

\ifshort\begin{theorem}[\appsymb]\else\begin{theorem}\fi\label{thm:eiden_fix}
For a given \equib~game $G = (N, \xxx, \uuu)$ and a set of desired strategy profiles $\ooo \subseteq \xxx$, the smallest $\delta \geq 0$ such that $(G, \ooo, \delta)$ is a positive instance of \pEHG\ can be identified in time $\bigO(|\ooo|\max_{i \in N}(n| O_i |^{|X_i \setminus O_i|}|\ooo||X_i \setminus O_i| + |O_i |^{n|X_i \setminus  O_i|}))$.
\end{theorem}

\ifshort
\appendixproof{thm:eiden_fix}{
\else
\begin{proof}%
\fi
By \cref{lem_fix_correct} Algorithm~\ref{alg:eidennew} correctly returns the minimum $\delta \geq 0$ such that there is a payment promise $\vvv$ for which $\max_{o \in O}\sum_{i \in N}\vvv_i(o) \leq \delta$ and $\xxx^{\star}_{G[\vvv]} \subset \ooo$. By \cref{theorem:false_promises}, we can translate this to an exact implementation of~$\ooo$ with cost $\delta$.

For a given $F_i, X_i, O_i$, \ComputeV iterates over $\ooo$ in lines \eqref{algli:comv2}-\eqref{algli:comv3}. The $\max$-operation on line \eqref{algli:comv4} iterates at worst over $X_i \setminus O_i$. Thus \ComputeV takes $\bigO(|\ooo||X_i \setminus O_i|)$ time.

Line \eqref{algli:initcount1} iterates over every player $i \in N$ and Line \eqref{algli:initcount2} every function from $X_i \setminus O_i \to O_i$, which there are $| O_i |^{|X_i \setminus O_i|}$. Therefore lines \eqref{algli:initcount1}-\eqref{algli:initcount3} take $\bigO(\max_{i \in N}n| O_i |^{|X_i \setminus O_i|}|\ooo||X_i \setminus O_i|)$ time.

Lines \eqref{algli:foreachF}-\eqref{algli:foreachFend} iterate over $\fff$ and performs a computation over $\ooo$. The time-complexity is thus $\bigO(\max_{i \in N}|O_i |^{n|X_i \setminus  O_i|}|\ooo|)$.

The remaining algorithm iterates over $\xxx$. Thus the total time complexity is $\bigO(|\ooo|\max_{i \in N}(n| O_i |^{|X_i \setminus O_i|}|\ooo||X_i \setminus O_i| + |O_i |^{n|X_i \setminus  O_i|}))$.
\ifshort
}
\else
\end{proof}
\fi

\section{Characterization of Cost-0 Implementation}
\label{sec:char}
\appendixsection{sec:char}
\newcommand{\gennashlong}{promise-Nash equilibrium}
\newcommand{\gennash}{PNE}
\looseness=-1
In this section we characterize rectangular strategy profiles $\mathcal{P} = P_1 \times P_2 \times \ldots \times P_n$ that can be implemented at zero cost.
We call such profiles \gennashlong~(\gennash).
The naming comes from two considerations.
First, if each player~$i$ has only one strategy in~$P_i$, then a \gennash\ is equivalent to a Nash equilibrium.
Second, a \gennash\ encapsulates the notion that no player $i$ has an incentive to switch towards a strategy outside of $P_i$ \emph{provided} that each other player $j$ plays only strategies in $P_j$.
We believe this notion to be of independent interest because it models situations in which certain types of strategies may be off limits for, e.g., moral reasons.
Using \gennash\ may thus enable studying the price (or value) of morality and similar ideas.
Formally:

\begin{definition}
  Let $G = (N, \mathcal{X}, \mathcal{U})$ be a game.
  A rectangular strategy profile region $P_1 \times \dots \times P_n \subset X_1 \times \dots \times X_n$ is a \emph{\gennashlong~(\gennash)} if
  \begin{multline*}
    \forall~i \in N, \forall~x_i \in X_i \setminus P_i~\exists~p_i \in P_i \colon \\
    \forall~p_{-i} \in P_{-i} \colon \quad U_i(p_i, p_{-i}) \geq U_i(x_i, p_{-i}).
\end{multline*}
\end{definition}

We observe the following.

\ifshort\begin{theorem}[\appsymb]\else\begin{theorem}\fi\label{thm:zero-cost-is-gen-nash}
  Let $G = (N, \mathcal{X}, \mathcal{U})$ be a game.
  A rectangular strategy profile region $\mathcal{P} = P_1 \times \dots \times P_n \subset X_1 \times \dots \times X_n$ can be implemented with cost 0 if and only if $\mathcal{P}$ is a \gennashlong.
\end{theorem}
In fact, this theorem has an interpretation in the morality setting mentioned above:
Let a profile $\mathcal{P}$ consist only of moral strategies and no amoral ones.
Then $\mathcal{P}$ is a \gennash\ if and only if morality is incentivized without any incentives having to actually be realized.
In other words, morality is self-enforcing if and only if it constitutes a \gennash.

\ifshort
\appendixproof{thm:zero-cost-is-gen-nash}{
\else
\begin{proof}[Proof of \cref{thm:zero-cost-is-gen-nash}]
\fi
  Assume first that $\mathcal{P}$ can be implemented at cost~0.
  Let $\mathcal{V}$ be a corresponding payment promise.
  We aim to show that $\mathcal{P}$ is a \gennash.
  Consider a player $i \in N$ and a strategy $x_i \in X_i \setminus P_i$.
  Since $\mathcal{V}$ implements $\mathcal{P}$, there is a strategy $x'_i$ that dominates $x_i$ in $G[\mathcal{V}]$.
  Furthermore, by transitivity of domination we may assume that $x'_i \in P_i$.
  Since the cost of $\mathcal{V}$ is 0 we have for each $p_{-i} \in \mathcal{P}_{-i}$ that $V_i(x'_i, p_{-i}) = 0$.
  However, since $x'_i$ dominates $x_i$ we have $[U_i + V_i](x'_i, p_{-i}) \geq [U_i + V_i](x_i, p_{-i})$.
  Thus, necessarily $U_i(x'_i, p_{-i}) \geq U_i(x_i, p_{-i})$.
  Thus, we may take $x'_i$ to be the strategy $p_i$ promised to exist by the definition of \gennash s.
  Thus, $\mathcal{P}$ indeed is a \gennash.

  Now assume that $\mathcal{P}$ is a \gennash.
  We define the following payment promise $\mathcal{V}$.
  For each $i \in N$, each $p_i \in P_i$ and each $x_{-i} \in \mathcal{X}_{-i} \setminus \mathcal{P}_{-i}$ we put $V_i(p_i, x_{-i}) = \infty$.
  All other values of $V_i$ are~0.
  Observe that all non-zero payment promises are for strategy profiles that include at least one strategy not in some $P_j$.
  Thus, if we can show that for each player $j \in N$ it is the case that each strategy $x_j \in X_j \setminus P_j$ is not undominated, then the cost of $\mathcal{V}$ is 0.
  
  Consider a player $i \in N$ and a strategy $x_i \in X_i \setminus P_i$.
  We claim that $x_i$ is not undominated.
  By the definition of a \gennash, there exists a strategy $p_i \in P_i$ such that $\forall~p_{-i} \in P_{-i} \colon U_i(p_i, p_{-i}) \geq U_i(x_i, p_{-i})$.
  By definition of $\mathcal{V}$ we thus have $[U_i + V_i](p_i, p_{-i}) \geq [U_i + V_i](x_i, p_{-i})$.
  Furthermore, by definition of $\mathcal{V}$ we have for each $x_{-i} \in \mathcal{X}_{-i} \setminus \mathcal{P}_{-i}$ that $[U_i + V_i](p_i, x_{-i}) = \infty > [U_i + V_i](x_i, x_{-i})$.
  Thus indeed $p_i$ dominates $x_i$.

  It follows that indeed the cost of $\mathcal{V}$ is 0.
  Furthermore, $\mathcal{V}$ implements $\mathcal{P}$ by the domination relations shown above.
\ifshort
} %
\else
\end{proof}
\fi

\section*{Acknowledgments}
Jiehua Chen and Sofia Simola are supported by the Vienna Science and Technology Fund (WWTF), grant number VRG18-012.
Manuel Sorge acknowledges funding by the Alexander von Humboldt Foundation.

\bibliography{literature}

\end{document}